\documentclass[11pt, a4paper]{article}
\usepackage[scaled]{helvet}
\usepackage [left,pagewise] {lineno}
\usepackage {microtype}
\usepackage {fullpage}
\usepackage {color}
\usepackage {amsmath}
\usepackage {amssymb}
\usepackage {amsthm}
\usepackage {euscript}
\usepackage {graphicx}
\usepackage {subfigure}
\usepackage {wrapfig}
\usepackage [rflt] {floatflt}
\usepackage {xspace}
\usepackage[medium,compact] {titlesec}
\usepackage {mdwlist, enumitem}
\usepackage {tabularx,booktabs}

\newtheorem {theorem} {Theorem}
\newtheorem {lemma} {Lemma}

\graphicspath{{figures/}}

\definecolor {infocolor} {rgb} {0.6,0.6,0.6}

\newcommand{\niceremark}[3]{\textcolor{blue}{\textsc{#1 #2:}} \textcolor{red}{\textsf{#3}}}

\newcommand{\rodrigo}[2][says]{\niceremark{Rodrigo}{#1}{#2}}

\newcommand{\maarten}[2][says]{\niceremark{Maarten}{#1}{#2}}

\DeclareMathOperator{\colvis}{ColVis}
\DeclareMathOperator{\vorvis}{VorVis}
\DeclareMathOperator{\vis}{Vis}
\newcommand{\viewshed}[2][\T]{\ensuremath{\mathcal{V}_{#1}(#2)}\xspace}
\newcommand{\vorviewshed}[3][\T]{\ensuremath{\mathcal{W}_{#1}(#2,#3)}\xspace}

\newcommand{\pow}[1]{\ensuremath{\mathit{pow}(#1)}\xspace}

\newcommand{\overviewshed}{\ensuremath{\mathbb{V}}\xspace}

\newcommand{\pts}{\ensuremath {\cal P}}

\newcommand{\vase}[2]{\ensuremath{\uparrow_{#1}^{#2}}\xspace}

\newtheorem{observation}{Observation}
\newtheorem{corollary}{Corollary}
\newtheorem{proposition}{Proposition}
\newtheorem{definition}{Definition}

\newcommand{\brep}[2]{\noindent{\bf #1~\ref{#2}.}\ \slshape}
\newcommand{\erep}{\normalfont}

\newcommand{\mkmcal}[1]{\ensuremath{\mathcal{#1}}\xspace}

\newcommand {\mathset} [1] {\ensuremath {\mathbb {#1}}}
\newcommand {\R} {\mathset {R}}

\newcommand{\B}{\mkmcal{B}}
\newcommand{\V}{\mkmcal{V}}
\newcommand{\G}{\mkmcal{P}}
\newcommand{\T}{\mkmcal{T}}

\newcommand{\A}{\mkmcal{A}}
\newcommand{\VD}{\mathrm{V\hspace{-2pt}D}}
\newcommand{\PD}{\mathrm{P\hspace{-1pt}D}}
\renewcommand{\S}{\G}

\begin{document}

\title{Terrain Visibility with Multiple Viewpoints\,\footnotemark[1]
}

\author{
 Ferran Hurtado\footnotemark[2]\and
 Maarten L{\"o}ffler\footnotemark[3] \and
 In{\^e}s Matos\footnotemark[4]\,\,\,\footnotemark[2] \and
 Vera Sacrist\'{a}n\footnotemark[2] \and
 Maria Saumell\footnotemark[5] \and
 Rodrigo I. Silveira\footnotemark[4]\,\,\,\footnotemark[2]  \and
 Frank Staals\footnotemark[3]
 }

\renewcommand{\thefootnote}{\fnsymbol{footnote}}

\footnotetext[1]{A preliminary version of this paper appeared in Proc. 24th International Symposium on Algorithms and Computation, pp. 317-327. The algorithm in Section~\ref{subs:alg-vis-map-1.5D} appeared in Abstracts from 30th European Workshop on Computational Geometry.}

\footnotetext[2]{Departament de Matem\`atica Aplicada II,
Universitat Polit\`ecnica de Catalunya, Barcelona, Spain, {\tt
\{ferran.hurtado,vera.sacristan,rodrigo.silveira\}@upc.edu}.}

\footnotetext[3]{Department of Information and Computing Sciences, Universiteit
  Utrecht, The Netherlands, {\tt \{m.loffler,f.staals\}@uu.nl}.}

\footnotetext[4]{Dept. de Matem\'atica \& CIDMA, Universidade de Aveiro, Portugal, {\tt \{ipmatos,rodrigo.silveira\}@ua.pt}.}

\footnotetext[5]{Department of Mathematics and European Centre of Excellence NTIS (New Technologies for the Information Society), University of West Bohemia, Czech Republic, {\tt saumell@kma.zcu.cz}.}

\renewcommand{\thefootnote}{\arabic{footnote}}

\maketitle

\begin{abstract}
We study the problem of visibility in polyhedral terrains in the presence
  of multiple viewpoints.  We consider a triangulated terrain with $m>1$
  viewpoints (or guards) located on the terrain surface. A point on the terrain
  is considered \emph{visible} if it has an unobstructed line of sight to at
  least one viewpoint. We study several natural and fundamental visibility
  structures: (1) the visibility map, which is a partition of the terrain into visible
  and invisible regions; (2) the \emph{colored} visibility map, which is a partition of
  the terrain into regions whose points have exactly the same visible
  viewpoints; and (3) the Voronoi visibility map, which is a partition of the terrain
  into regions whose points have the same closest visible viewpoint. 
	We study the complexity of each structure for both 1.5D and 2.5D terrains, and
  provide efficient algorithms to construct them.
  Our algorithm for the visibility map in 2.5D terrains improves on the only existing algorithm in this setting. To the best of our knowledge, the other structures have not been studied before.
\end{abstract}

\section{Introduction}
\label{sec:intro}

\emph{Visibility} is one of the most studied topics in computational geometry. Many different terms, like art-gallery problems, guarding, or visibility itself, have been used during the last three decades to refer to problems related to the question of whether two objects are \emph{visible} from each other, amidst a number of obstacles.

In this paper we are interested in visibility on \emph{terrains}.  A \emph{2.5D terrain} is an $xy$-monotone polyhedral surface in $\mathbb{R}^3$.  We also study 1.5D terrains: $x$-monotone polygonal lines in $\mathbb{R}^2$. The obstacles we consider are the terrain edges or triangles themselves. A fundamental aspect of visibility in terrains is the \emph{viewshed} of a point (i.e. the \emph{viewpoint}): the (maximal) regions of the terrain that the viewpoint can see (see Figure~\ref{fig:viewshed}).

In a 1.5D terrain, the viewshed is almost equivalent to the \emph{visibility
  polygon} of the viewpoint, since the terrain can be turned into a simple
polygon by ``closing up'' the unbounded space above the terrain.  Therefore
well-known linear-time algorithms to construct visibility polygons can be
applied (e.g.~\cite{js-clvpa-87}).
In
2.5D the viewshed is more complex.  In an $n$-vertex terrain, the
viewshed of a viewpoint can have $\Theta(n^2)$ complexity.  The
best algorithms known to compute it take $O((n + k) \log n \log \log n)$ time~\cite{rs-eoshsra-88}, and $O((n
\alpha(n) + k) \log n)$ time~\cite{kos-ehsrosus-92}, where $k$ is the size of the
resulting viewshed, and $\alpha(n)$ is the extremely slowly
growing inverse of the Ackermann function.
 I/O-efficient versions of this problem for grid terrains have also been studied recently
(e.g.~\cite{ammfc-evc-11,fht-ivcmgt-09,htz-cvtem-08}), as well as other terrain
visibility structures like \emph{horizons} and
\emph{offsets}~\cite{FM03}.

While the computation of the viewshed from one viewpoint on a terrain has been
thoroughly studied, it is surprising that a natural and important variant has
been left open: What happens if instead of \emph{one} single viewpoint, one has
\emph{many}, say $m>1$, different viewpoints on the terrain? The \emph{common
  viewshed}, or \emph{visibility map} can then be defined as the regions of the
terrain that can be seen from \emph{at least one} viewpoint. Computing the
viewshed from each single viewpoint and then taking the union of the $m$
viewsheds is a straightforward way to compute it, but it
has a high running time that does
not take the final size of the visibility map into account. Obtaining more
efficient algorithms for this and other related problems is the main focus of
this paper.

To the best of our knowledge, there are no other studies on (the complexity of) the visibility map
of multiple viewpoints. 
The only related work that we are aware of for 1.5D terrains deals with the problem of determining whether at least two viewpoints \emph{above} a 1.5D terrain can see each other, for which Ben-Moshe et al.~\cite{bhkm-cvgpp-04} present an $O((n+m)\log m)$-time algorithm.
For 2.5D terrains we can only mention the work
by Coll et al.~\cite{cms-gvmv-10}, who essentially overlay the $m$
individual viewsheds without studying the complexity of the
visibility map. This results in the very high running time of
$O(m^2n^4)$. In
addition, a few papers deal with the computation of viewsheds for
multiple
viewpoints for rasterized terrains 
\cite{cfms-mvmtt-07,cms-gvmv-10,lzl-mostsa-06}.
We refer to the survey by De
Floriani and Magillo~\cite{FM03} for an overview of terrain visibility from the Geographic Information Science (GIS) perspective.
A recent paper by Lebeck et al.~\cite{lma-chopt-13} studies finding paths  on 2.5D  terrains that are likely to be not visible from a ``reasonable'' (but unknown) set of viewpoints.

We would like to highlight the fact that it is not due to its lack
of interest that visibility from multiple viewpoints has been
overlooked up to now. Visibility in 1.5D terrains has been
thoroughly studied from related perspectives, and in particular
the problem of \emph{placing} a minimum number of viewpoints to
cover a 1.5D terrain has received a lot of attention
(e.g.~\cite{benmoche2007constantfactor,Clarkson:2007,ElbassioniKMMS11,gibson2009approximation,King:2011,Nilsson94}).
Their theoretical interest and the fact that 1.5D terrains already
pose a difficult challenge are the main motivation behind our work
in that dimension.

For 2.5D terrains, the study of visibility from multiple viewpoints has a large number of applications in GIS.
Here we only present a few concrete examples.
 For early fire prevention, fire lookout towers are essential. Thus visibility studies are crucial to determine the total area visible from a set of towers, and therefore evaluate the effectiveness of fire detection systems (e.g.~\cite{crsar-fpp-07}).
 In visual pollution studies, for instance when considering where to install a wind farm, turbines should be as hidden as possible, and should not be visible from certain ``sensitive sites'' (like touristic points and major buildings). Visibility analysis has been successfully used to identify suitable installation areas that are not visible from sensitive sites~\cite{m-cwpl-06}.
 Finally, even though in this paper we use the term \emph{visibility}, our results also apply to other contexts.
 For example, in sensor networks visibility problems are formulated as \emph{coverage problems}.
 This type of monitoring can be performed by radars, antennas, routers and basically any device that is able to send or receive some sort of wireless signal. Each of these devices can be placed almost anywhere on a terrain, so coverage is the discipline that measures the quality of the chosen device placement scheme. 
Numerous problems related to visibility arise in this area.
We refer to the survey by Yick et al.~\cite{ymg-wsns-08} for a more detailed treatment of the subject.

\paragraph{Problem Statement.} A 2.5D terrain \T consists of $n$ vertices, $O(n)$ edges, and $O(n)$ faces. Let
$V(\T)$ denote the set of vertices, $E(\T)$ the set of edges, and $F(\T)$ the
set of faces of \T. A 1.5D terrain \T consists of $n$ vertices and $n-1$ edges. We
again denote the vertices of \T with $V(\T)$ and the edges of \T with $E(\T)$.
With some slight abuse of notation, we sometimes use \T to refer to the \emph{set} of points on the surface of \T.

\begin{figure}[t]
  \centering
  \includegraphics[page=1]{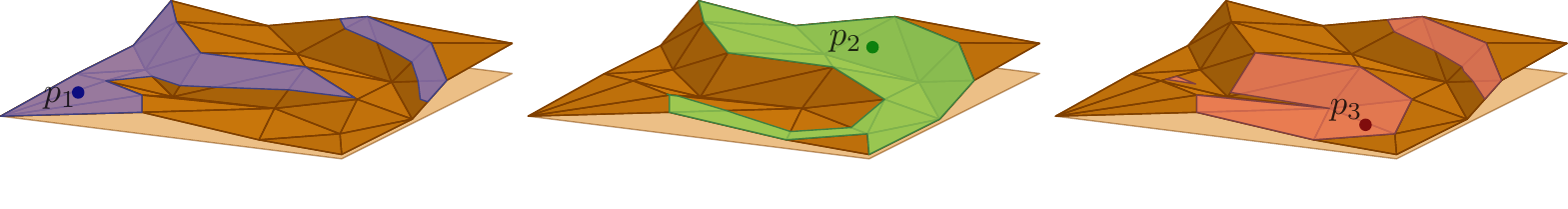}
  \caption{The viewsheds of three viewpoints on a 2.5D terrain.}
  \label{fig:viewshed}
\end{figure}
\begin{figure}[t]
  \centering
  \includegraphics[page=2]{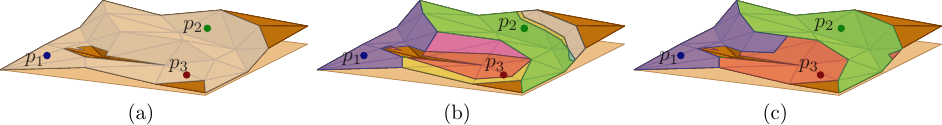}
  \caption{The visibility map (a), the colored visibility map (b), and the
    Voronoi visibility map~(c).}
  \label{fig:visibility_structures}
\end{figure}

For any point $p$ on the terrain \T (either a 2.5D terrain or a
1.5D terrain), the \emph{viewshed} of $p$ on \T, denoted by
\viewshed{p}, is the maximal set of points on \T that are
\emph{visible} from $p$. A point $q$ is visible from $p$ if and
only if the line segment $\overline{pq}$ does not contain any
point strictly below the terrain surface. 
Note that our definition of visibility is symmetric,
and that viewpoints have unlimited sight
(Section~\ref{sec:limited-sight} discusses the situation of
viewpoints with limited sight).
The viewshed
\viewshed{\G} of a set of viewpoints \G is the set of points
visible from at least one viewpoint in \G, that is, $\viewshed{\G}
= \bigcup_{p
  \in \G} \viewshed{p}$.

\newcommand{\closest}[2][\T]{\ensuremath{\mathit{closest}_{#1}(#2)}}

Given a set of viewpoints \G, we define the \emph{Voronoi viewshed}
\vorviewshed{p}{\G} of a viewpoint $p \in \G$ as the set of points in the viewshed of
$p$ that are closer to $p$ than to any other viewpoint that can see them. More
precisely, $\vorviewshed{p}{\G} = \viewshed{p} \cap \{x \mid x \in \T \land
\closest{x,\G} = p \}$, where $\closest{x,\G} $ denotes the closest (in terms of the Euclidean distance) viewpoint in \G that can see a point $x$ on \T. 

We can now formally define the three visibility structures studied in this paper, illustrated in Figure~\ref{fig:visibility_structures}.

\begin{definition}
  The \emph{visibility map} $\vis (\T,\G)$ is a subdivision of the terrain \T
  into a visible region $ \viewshed{\G}$ and an invisible region $\T
  \setminus  \viewshed{\G}$.
  \end{definition}

\begin{definition}
  The \emph{colored visibility map} $\colvis (\T,\S)$ is a subdivision of the
  terrain \T into maximally connected regions $R$, each of which is covered by exactly
  the same subset of viewpoints $\G' \subseteq \G$.
Each region $R$ is a
  (maximally connected) subset of $\bigcap_{p \in \G'} \viewshed{p}$ and we
  have that $R \cap \bigcup_{p \in \G\setminus\G'} \viewshed{p} =
  \emptyset$.
\end{definition}

\begin{definition}
  The \emph{Voronoi visibility map} $\vorvis (\T,\S)$ is a subdivision of the
  terrain \T into maximally connected regions, each of which is a subset of
  the Voronoi viewshed \vorviewshed{p}{\G} of a viewpoint $p \in \G$.
\end{definition}

We denote the \emph{size}, that is, the total complexity of all its regions, of
$\vis(\T,\S)$, $\colvis(\T,\S)$, and $\vorvis(\T,\S)$, by $k$, $k_c$, and
$k_v$, respectively.

In the remainder of the paper we assume that \G is a set of $m$ viewpoints on
the terrain surface. For simplicity, we assume that the viewpoints are placed
on terrain vertices, which implies that $m \leq n$. We consider this a reasonable assumption, since in the applications that motivate this work the number of terrain vertices is considerably larger than the number of viewpoints. Additionally, for the case of 1.5D terrains, if three or more viewpoints lie on the same edge, then the union of the viewsheds of the leftmost and rightmost viewpoints contains the viewshed of any other viewpoint on the edge (this partially follows from Observation~\ref{obs:left-vis}). Therefore, in general, if the initial number of viewpoints is large, in most applications the number of relevant viewpoints is still $O(n)$. 

\paragraph{Results.}

We present a comprehensive study of the visibility structures defined above.
We analyze the complexity of all the structures and propose algorithms to compute them.
Our results 
for unlimited sight 
are summarized in Table~\ref{tbl:results}. 

Regarding 1.5D terrains,
all our algorithms avoid computing individual viewsheds.  $\vis (\T,\G)$ is
computed in nearly optimal running time, while the algorithms for $\colvis
(\T,\G)$ and $\vorvis (\T,\G)$ are output-sensitive, although the running time of the latter depends also on~$k_c$. 

As for 2.5D terrains, we prove 
that the
maximum complexity of $\vis (\T,\G)$ and $\colvis (\T,\G)$ is much less than the overlay
of the viewsheds, as implicitly assumed in previous work~\cite{cms-gvmv-10}.
Using that, we show how a combination of well-known algorithms can be used to compute the visibility structures reasonably fast.

Finally, in Section~\ref{sec:limited-sight} we analyze how our results change when viewpoints have limited sight, which is a very common situation for practical purposes. 

\begin{table}[t]
  \centering
  \begin{tabularx}{\textwidth}{rccX<{\centering}}
    \toprule
    \multicolumn{4}{c}{1.5D Terrains} \\
    Structure   & & Max. size        & Computation time                                     \\
    \cmidrule(l){1-1}                 \cmidrule(l){3-4}
      $\vis$    & & $\Theta(n) $     & $O(n + m \log m)$                                    \\
      $\colvis$ & & $\Theta(mn) $    & $O(n + (m^2 + k_c) \log n)$                          \\
      $\vorvis$ & & $\Theta(mn) $    & $O(n + (m^2 + k_c) \log n + k_v (m+\log n \log m))$  \\
    \\[1.2\baselineskip]
    \toprule
    \multicolumn{4}{c}{2.5D Terrains} \\
    Structure   & & Max. size        & Computation time                                    \\
    \cmidrule(l){1-1}                 \cmidrule(l){3-4}
      $\vis$    & & $\Theta(m^2n^2)$ & $O(m(n\alpha(n) + \min(k_c,n^2))\log n + mk_c)$      \\
      $\colvis$ & & $\Theta(m^2n^2)$ & $O(m(n\alpha(n) + \min(k_c,n^2))\log n + mk_c)$      \\
      $\vorvis$ & & $O(m^3n^2)$      & $O(m(n\alpha(n) + \min(k_c,n^2))\log n + mk_c\log m)$\\
    \bottomrule
  \end{tabularx}
  \vspace*{0.1cc} \caption{Complexity and computation time of the three visibility structures,
	 for unlimited sight,
 for an $n$-vertex terrain with $m$ viewpoints.
}
  \label{tbl:results}
\end{table}

\section{1.5D Terrains}

A 1.5D terrain \T is an $x$-monotone polygonal chain, specified by its sorted list of $n$ vertices.
In addition, we are given a set \S of $m$ viewpoints placed on some of the terrain vertices.
We first study the complexity of the visibility, colored visibility, and Voronoi visibility maps. Then we present algorithms to compute them efficiently. Note that in 1.5D our visibility structures can be seen as subdivisions of the $x$-axis
into intervals.

In this section we use the following notation.
We denote the $x$- and $y$-coordinates of a point $p \in \mathbb{R}^2$ by $x(p)$ and $y(p)$, respectively.
We use $\T[a,c]$, for $a,c$ in
$\T$ and $x(a) < x(c)$, to denote the closed portion of the
terrain between $a$ and $c$, and $\T(a,c)$ for the open portion.
Similarly, if $q \in \mathbb{R}^2$, we use $\T[q]$ to denote the point on \T whose
$x$-coordinate is equal to $x(q)$; if $q \in \mathbb{R}$, $\T[q]$ denotes the point on \T whose
$x$-coordinate is equal to $q$. 

\subsection{Complexity of the Visibility Structures}

We start by considering the complexity of the visibility and
colored visibility maps.

\begin{theorem}
The visibility map  $\vis (\T,\G)$, for a 1.5D terrain \T, has
complexity $\Theta(n)$.
\end{theorem}

\begin{proof}
There are two types of points of \T that contribute to the
complexity of $\vis (\T,\G)$: vertices of \T, and points where the
terrain changes between visible and invisible (notice that
there could be points that are of both types). There are $n$ points of the first type, and in the following paragraph we show that the points of the second type amount to $O(n)$. Consequently, $k$ is
$\Theta(n)$.

Consider an edge $e\in E(\T)$ and a viewpoint $p_i$ to the left of
$e$. If $p_i$ sees some interior point $q$ of $e$, then $p_i$ also
sees all the segment from $q$ to the rightmost endpoint of $e$. A
symmetric situation occurs if $p_i$ is to the right of $e$. Thus
there are three possible situations for the visibility of $e$ with respect to the whole set $\G$: 
(i) $e$ is fully visible or invisible,
(ii) $e$ contains two connected portions, including
endpoints, where one is visible and the other is not, 
or (iii) $e$ contains one invisible connected portion between two visible ones.
This last case is shown in Figure~\ref{fig:vis-edge}. In any case,
the interior of $e$ contains at most two points in which the terrain
changes between visible and invisible.
\end{proof}

\begin{figure}[t]
\centering
\subfigure[] 
{
    \includegraphics[scale=1]{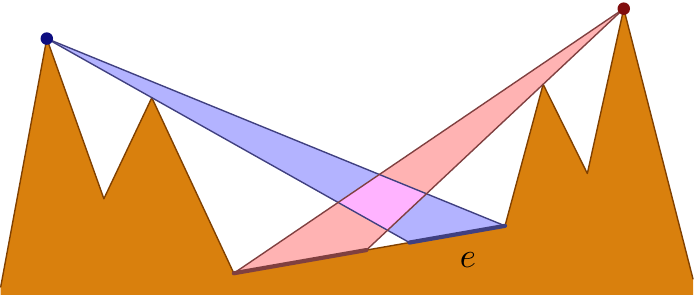}
  \label{fig:vis-edge}
} \hspace{1cm}
\subfigure[] 
{
    \includegraphics[scale=1]{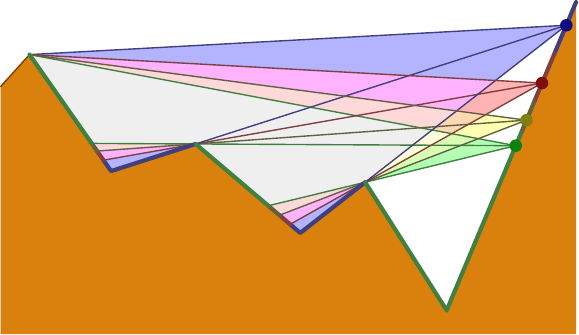}
  \label{fig:Quadratic}
}
\caption{(a) Edge $e$ contains one invisible connected portion
between two visible ones. (b) Every other edge has four different
regions of $\colvis(\T,\S)$ and four different regions of $\vorvis(\T,\S)$.
Viewpoints are indicated with disks.
}
\end{figure}

\begin{theorem}
The colored visibility map $\colvis(\T,\S)$, for a 1.5D terrain \T,
has maximum
 complexity $\Theta(mn)$.
\end{theorem}

\begin{proof}
As already mentioned, once a viewpoint sees a given point $q$
on an edge $e\in E(\T)$, it necessarily sees the whole segment from
$q$ to one of $e$'s endpoints. Hence, all $m$ viewpoints may produce
at most $m+1$ different regions of $\colvis(\T,\S)$ on $e$.
Therefore the complexity of $\colvis(\T,\S)$ is $O(mn)$. The
example in Figure~\ref{fig:Quadratic} shows that the bound is tight.
\end{proof}

Next we study the Voronoi visibility map. We first consider this
map restricted to an edge $e\in E(\T)$. We define $\S_l$
(respectively, $\S_r$) as the subset of \S containing the
viewpoints on the left (respectively, on the right) of $e$.
If a
viewpoint is placed at the leftmost (respectively, rightmost)
endpoint of $e$, then we assign it to $\S_l$ (respectively,
$\S_r$).
We define $m_l=|\S_l|$ and $m_r=|\S_r|$.

\begin{lemma} \label{lem:vor-edg}
The Voronoi visibility map $\vorvis(\T,\S)$ restricted to an edge $e\in E(\T)$ has
at most $4m-2$ regions.
\end{lemma}

\begin{proof}
Without loss of generality, we assume that $e$ has a positive
slope. We first show that $\vorvis(\T,\S_l)$ restricted to $e$ has
at most $2m_l$ regions. Suppose that we traverse $e$ from bottom
to top, and that at point $q$ on $e$ we exit
$\vorviewshed{p_i}{\S_l}$ and enter $\vorviewshed{p_j}{\S_l}$.
Then either $p_j$ becomes visible at $q$, or the bisector of $p_i$
and $p_j$ intersects $e$ at $q$. We say that $q$ is of type I in
the first case, and of type II in the second case. Since the
portion of $e$ seen by any viewpoint is connected, there are at
most $m_l$ points of type I on $e$.

In order to bound the number of type II points, suppose that at point $q$ we exit
$\vorviewshed{p_i}{\S_l}$ and enter $\vorviewshed{p_j}{\S_l}$ due
to an event of type II. Notice that both $p_i$ and $p_j$ see the
interval of $e$ from $q$ up to the topmost point of $e$.
Furthermore, viewpoint $p_j$ is closer than $p_i$ to all the
points on such interval, since the bisector of $p_i$ and $p_j$
intersects $e$ only once at $q$. Hence no other component of
$\vorviewshed{p_i}{\S_l}$ will be found on $e$. In particular,
there will be no other events produced by the intersection of $e$
with the bisector of $p_i$ and some other viewpoint of $\S_l$.
Thus there are at most $m_l-1$ points of the second type, which
means that in $e$ there are at most $2m_l-1$ points where the regions
of $\vorvis(\T,\S_l)$ change.

Analogous arguments show that $\vorvis(\T,\S_r)$ restricted to $e$
has at most $2m_r-1$ points where the regions change. 
In order to count the number of components in $\vorvis(\T,\S)$, consider the $2m-2$ points where there is a change of region in $\vorvis(\T,\S_l)$ or $\vorvis(\T,\S_r)$, and suppose that we traverse $e$ stopping at each of these points.
Take a portion of $e$ between two consecutive points, and let
$p_i$ be the viewpoint such that the interval belongs to
$\vorviewshed{p_i}{\S_l}$, and $p_j$ be the viewpoint such that
the interval belongs to $\vorviewshed{p_j}{\S_r}$. If the bisector
of $p_i$ and $p_j$ does not intersect this portion of $e$, then
the portion is a component of either $\vorviewshed{p_i}{\S}$ or
$\vorviewshed{p_j}{\S}$. If the bisector of $p_i$ and $p_j$
intersects this portion of $e$, then the portion is divided into a
component of $\vorviewshed{p_i}{\S}$ and a component of
$\vorviewshed{p_j}{\S}$. The latter case shows that each of these
$2m-1$ intervals can hold two different components and
consequently $\vorvis(\T,\S)$ restricted to $e$ can have up to
$4m-2$ regions.
\end{proof}

\begin{theorem}
The Voronoi visibility map $\vorvis(\T,\S)$, for a 1.5D terrain \T,
has maximum complexity $\Theta(mn)$.
\end{theorem}

\begin{proof}
The upper bound follows from Lemma~\ref{lem:vor-edg}. The
lower bound is achieved by a configuration of viewpoints on a
particular terrain \T that can be repeated so that that every
other edge of \T has as many Voronoi regions as viewpoints, for
arbitrary $n$ and $m$. An example 
is shown in
Figure~\ref{fig:Quadratic}.
\end{proof}

Further, we note that the three maps have complexity $\Omega(n)$ because the vertices of
$\T$ contribute to their complexity. Notice that if we are
interested in the output as a subdivision of the domain, rather
than the terrain, we do not necessarily require $\Omega (n)$
space, since it suffices to provide the endpoints of each region.

Additionally, it holds that $k\le \min(k_c,k_v)$, and the
configuration of Figure~\ref{fig:Quadratic} shows that $k$ can be
$\Theta(n)$ while both $k_c$ and $k_v$ are $\Theta(mn)$. On the
other hand, it is not difficult to produce examples showing that
the complexity of $\vorvis(\T,\S)$ can be higher than, lower than,
or equal to that of $\colvis (\T,\S)$.

\subsection{Algorithms to Construct the Visibility Structures}

\subsubsection{Computing the Visibility Map} \label{subs:alg-vis-map-1.5D}

To construct the visibility map we first compute the left- and right-visibility maps, and then merge them. The
\emph{left-visibility map} partitions \T into two regions: the
visible and the ``invisible'' portions of the terrain, where
\emph{visible} means visible from a viewpoint \emph{on} or
\emph{to the left} of that point of the terrain (thus in the
left-visibility map, viewpoints only see \emph{themselves} and
portions of the terrain \emph{to their right}). The
right-visibility map is defined analogously.
Next we present the construction of the left-visibility map (thus,
\emph{visible} stands for \emph{left-visible}).

\paragraph{Some properties of visibility in 1.5D terrains.} 

We say that a viewpoint $p_1$ \emph{dominates} another viewpoint $p_2$ at a given $x$-coordinate $x_1$ if for all $p \in \T$ with $x(p) \geq x_1$ it holds that: if $p_2$ sees $p$, then $p_1$ also sees $p$.

The algorithm uses a couple of consequences
of the so-called \emph{order claim}:

\begin{lemma}[Claim~2.1 in \cite{benmoche2007constantfactor}]
Let $a,b,c,$ and $d$ be four points on $\T$ such that
$x(a)<x(b)<x(c)<x(d)$. If $a$ sees $c$ and $b$ sees $d$, then $a$
sees $d$.
\end{lemma}

The following observation is a direct consequence of the order claim.

\begin{observation} \label{obs:left-vis}
Let $q \in \T$ be a point visible from viewpoints
$p_i$ and $p_j$, with $p_i$ to the left of $p_j$. For any 
$w\in \T$ to the right of $q$, if $p_i$ does not see $w$, then
$p_j$ does not see $w$ either (i.e. $p_i$ dominates $p_j$ at $x(q)$).
\end{observation}

It is convenient to introduce some terminology related to rays.
Given a viewpoint $p_i$ and a vertex $v_k \in \T$, the ray with origin $p_i$ and direction $\overrightarrow{p_iv_k}$ is called a \emph{shadow ray} if:
(i) $p_i$ sees $v_k$;
(ii) $p_i$ does \emph{not} see the points of \T
immediately to the right of $v_k$.
We use $\rho(p_i, v_k)$ to denote such ray.

The next corollary is illustrated in
Figure~\ref{fig:crossing-rays}.

 \begin{figure}[tb]
\centering
\includegraphics{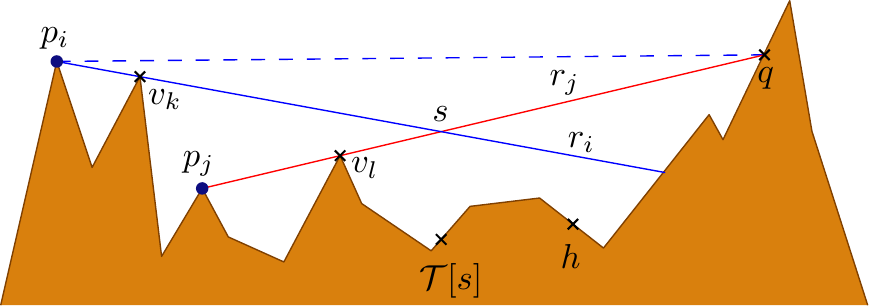}
\caption{Situation in Corollary~\ref{cor:crossing}.}
\label{fig:crossing-rays}
\end{figure}

\begin{corollary} \label{cor:crossing}
Let $r_i=\rho(p_i, v_k)$ and $r_j=\rho(p_j, v_l)$ be two shadow rays with $x(v_k)<x(v_l)$.
Suppose there exists some $h\in \T$ such that $\T(v_k,h)$ lies below
$r_i$ and $\T(v_l,h)$ lies below $r_j$.
Suppose that $r_i$ and $r_j$ cross at a point $s$ such that
$x(v_l)<x(s)<x(h)$. 
Then, for any point $w\in \T$ to the right
of $\T[s]$, if $p_i$ does not see $w$, then $p_j$ does not see
$w$ either (i.e. $p_i$ dominates $p_j$ at $x(s)$).
\end{corollary}


\begin{proof}
We start by showing that $p_i$ is to the left of $p_j$. We first
notice that $r_j$ lies below $r_i$ in the region of the plane to
the left of the vertical line $x=x(s)$: Otherwise, the point $v_l$
would lie above the ray $r_i$, contradicting the fact that
$\T(v_k,h)$ lies below $r_i$. Additionally, we have
$x(v_k)<x(p_j)$, since otherwise the point $v_k$ would lie in $\T(p_j,v_l)$ and above the ray $r_j$, blocking the visibility between $p_j$ and $v_l$.
Therefore, we conclude $x(p_i)<x(v_k)<x(p_j)$ and we are in the
situation illustrated in Figure~\ref{fig:crossing-rays}.

If $p_j$ does not see any point $q\in \T$ to the right of
$\T[s]$, the result follows. Otherwise, let $q$ be the leftmost
point in $\T$ to the right of $\T[s]$ visible from $p_j$.
$\T(p_i,q)$ lies on or below the polygonal line $p_isq$, which,
except for its endpoints, lies below the segment $p_iq$. Hence,
$q$ is visible from $p_i$. We now apply
Observation~\ref{obs:left-vis} and conclude.
\end{proof}

For a fixed $x$-coordinate, each viewpoint can have at most one shadow ray as in Corollary~\ref{cor:crossing}:

\begin{observation}\label{obs:rays}
Given a point $h \in \T$ and a viewpoint $p_i$, there is at most one vertex $v_k$ such that $\T(v_k, h)$ is below $\rho(p_i,v_k)$.
\end{observation}

Therefore for simplicity we use $r_i$ to denote such ray for viewpoint $p_i$, assuming it exists.

\paragraph{Description of the algorithm.} The algorithm sweeps the terrain from left to right while
maintaining some information.  
Most notably, we maintain a set $L$ of rays 
corresponding to a (sub)set of viewpoints that are not visible at the
moment\footnote{Often in this section we use \emph{visible} to refer to visibility from the intersection of the sweep line and the terrain.} (possibly, $L=\emptyset$). More precisely, the set $L$ at $x=x(h)$ (for $h \in
\T$) is defined as

$L = \{ \rho(p_i, v_k) $ such that $T(v_k, h)$ lies below  $\rho(p_i, v_k)$ and $p_i$ has not been detected to be dominated at $x=x(h)\}$.


The general idea is: If, while sweeping through $h \in
\T$, we have $r_i=\rho(p_i, v_k)\in L$, then $\T(v_k,h)$ lies below
$r_i$. If the terrain crosses $r_i$ at some point
to the right of $h$, then it becomes visible from
$p_i$, so we are interested in detecting such intersection between $r_i$ and the terrain.
  
In order to decrease the size of $L$ and obtain a better running time, $L$ does not contain rays associated to viewpoints that, at some previous point of the sweep, have been detected to be dominated. In general, there are several distinct situations where one viewpoint dominates another; our algorithm only detects the ones described in Observation~\ref{obs:left-vis} and Corollary~\ref{cor:crossing}, but this is enough to achieve a running time of $O(n+m\log m)$.


During the sweep we also maintain some special viewpoints. If the
terrain is currently visible, we keep the leftmost visible
viewpoint, which we call \emph{primary} viewpoint and denote by
$p_a$. If the terrain is not visible, we maintain another
viewpoint, which we call \emph{secondary} viewpoint and denote by
$p_b$. The secondary viewpoint belongs to the set $\{ p_i\, |\,
r_i\in L\}$, and it is the viewpoint that is \emph{more likely} to
become visible around the portion of the terrain that we are
examining. More precisely, if we are sweeping through point $h \in \T$,
then $r_b$ is defined as the lowest ray in $L$ at $x=x(h)$. When the
terrain is not visible, we define $L'=L\backslash \{r_b\}$.
Otherwise, we set $L'=L$. See Figure~\ref{fig:second-view} for an illustration.

At any moment of the sweep, we know the
lowermost ray in $L'$. As we will see, the lowermost ray in $L$ (that is, the ray associated with the secondary viewpoint) can change $\Theta(n)$ times, while the lowermost ray in $L'$ can only change $O(m)$ times. Since handling these events takes logarithmic time, it turns out that maintaining the lowermost
ray in $L'$ instead of the lowermost ray $L$ is key for 
achieving a running time of $O(n+m\log m)$.

\begin{figure}[tb]
\centering
\includegraphics{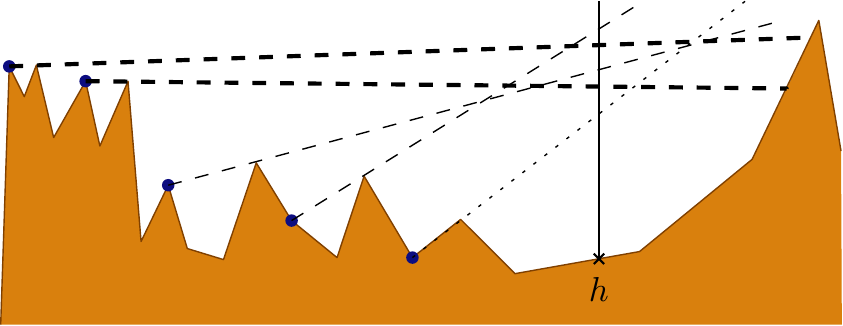}
\caption{When sweeping point $h$, the secondary viewpoint is the one whose ray is dotted. $L'$ contains only the thicker dashed rays (since, by Corollary~\ref{cor:crossing}, the viewpoints of the other ones are dominated).}
\label{fig:second-view}
\end{figure}

Finally, we maintain a boolean $\nu$ indicating whether the terrain
is currently visible or not.

The algorithm begins at the leftmost vertex, and starts sweeping the terrain as
explained below. For the sake of simplicity, in the following
description we assume that we know at any time which is the
lowermost ray in $L'$ and, if $\nu$ = \textsc{False}, which is the secondary
viewpoint. We will explain how to keep this information updated
later.

Initially, $L=L'=\emptyset$, $\nu$ = \textsc{False}, and $p_a=p_b=\bot$. Then the sweep begins at the leftmost vertex and stops at three types of events: (i) terrain
vertices; (ii) intersections between $r_b$ and the lowermost ray in $L'$ (in
which case the secondary viewpoint changes); (iii) intersections between the
lowermost ray in $L'$ and the second lowermost (in which case the lowermost ray
in $L'$ changes). Event (ii) is treated as follows: Suppose that we are about
to process edge $wv$ and we detect, say, that the secondary viewpoint changes
at $x=\alpha$, where $x(w)<\alpha<x(v)$, because $r_b$ intersects the lowermost
ray in $L'$ at $x=\alpha$.  Then we simply subdivide $wv$ into $w\T[\alpha]$
and $T[\alpha]v$ (that is, we ``add" a vertex to the terrain) and run two
iterations of the algorithm below, each with the appropriate secondary
viewpoint. Event (iii) is treated analogously. Therefore, in the explanation
below we assume that all iterations of the algorithm start when we encounter a
new vertex of the terrain.

We now explain an iteration of the algorithm. Let $w$ be the
vertex preceding $v$ in \T. We treat the interior of the edge $wv$
and the vertex $v$ separately.

\textbf{Detecting visibility changes in the interior of the edge
$wv$.} We distinguish several cases:
\begin{itemize} 
  \item[(i.1)] $\nu$ = \textsc{False} and $L=\emptyset$. We do nothing.
  \item[(i.2)] $\nu$ = \textsc{False} and $L\neq \emptyset$. We check whether the ray $r_b$ intersects the edge
  $wv$.
  In the affirmative, we compute the point of intersection and we set $\nu$ = \textsc{True} at that point. 
The viewpoint that was secondary, $p_b$, becomes the
  primary viewpoint, and $r_b$ is removed from $L$.
  We continue as in (i.3) or (i.4), depending on $L$ being empty or not.
 \item[(i.3)] $\nu$ = \textsc{True} and $L=\emptyset$. We do nothing.
 \item[(i.4)] $\nu$ = \textsc{True} and $L\neq \emptyset$. We check whether the lowermost
 ray $r_j$ 
 of $L'$ at $[x(w),x(v)]$ intersects the edge $wv$. If it does, we remove $r_j$ from $L'$ and find the new lowermost
ray of $L'$. Additionally, if $p_j$ is to the left of $p_a$, we
set $p_a=p_j$. We continue as in (i.3) or (i.4), depending on $L$
being empty or not.
\end{itemize}

\textbf{Dealing with the vertex $v$.}
 Let us first suppose that no viewpoint lies on $v$. We distinguish the following
 cases:
 \begin{itemize}
 \item[(ii.1)] $\nu$ = \textsc{False}. We do nothing.
 \item[(ii.2)] $\nu$ = \textsc{True} and $p_a$ continues being visible right after $v$. We do nothing.
 \item[(ii.3)] $\nu$ = \textsc{True} and $p_a$ stops being visible right after $v$. Viewpoint $p_a$ becomes the secondary viewpoint, and we add $\rho(p_a, v)$ to $L$. Additionally, we set
 $\nu$ = \textsc{False}.
 \end{itemize}

On the contrary, suppose that a viewpoint $p_i$ lies on $v$:
 \begin{itemize}
 \item[(ii.4)] $\nu$ = \textsc{False}. We set $\nu$ = \textsc{True} and $p_a=p_i$. If $L\neq \emptyset$, then $r_b$ is added to $L'$ and
the lowermost ray of $L'$ becomes $r_b$.
 There is no longer a
 secondary viewpoint.
 \item[(ii.5)] $\nu$ = \textsc{True} and $p_a$ continues being visible right after $v$. We do
 nothing.
 \item[(ii.6)] $\nu$ = \textsc{True} and $p_a$ stops being visible right after $v$.  We add $\rho(p_a, v)$ to $L'$, and update the
lowermost ray of $L'$. We set $p_a=p_i$.
 \end{itemize}

At this point, it only remains to explain the way we maintain the lowermost ray in $L'$ and, if $\nu$ = \textsc{False}, the secondary
viewpoint. We start with the former. 

\textbf{Maintaining the lowermost ray in $L'$.}
Knowing the lowermost ray in $L'$ during the whole sweep is equivalent to maintaining the lower envelope of $L'$.
We use a modification of Bentley-Ottmann's algorithm for line-segment intersections~\cite{bo-arcgi-79}, run on the set of rays that at some point belong to $L'$.
The algorithm essentially computes the intersections between the rays in $L'$ as the terrain is swept.
The sweep line status structure allows to retrieve the lowest ray in $L'$ at any time.  

The sweep line and the event queue are implemented using the standard data structures (i.e. a binary search tree and a priority queue).\footnote{In our case the sweep line could be represented with a simpler structure, like a doubly-linked list, but this would not affect the overall running time.} Next we argue that the overall running time of the sweep is only $O(m \log m)$.

First note that, by Observation~\ref{obs:rays}, at any moment of the sweep, $L'$ contains $O(m)$ rays.
Moreover, by Corollary~\ref{cor:crossing}, every time the sweep line goes through the intersection of two rays, the one corresponding to the viewpoint more to the right becomes dominated by the other one, so that ray will not be in $L'$ from that moment on, and what is even more important, no ray from that viewpoint will.
Thus the total number of intersections considered by the algorithm is $O(m)$.


The other types of events are insertions and deletions of rays. 
In total, we make at most $m$ insertions to $L'$: Indeed, we only add a ray in cases (ii.4) or (ii.6), and we can charge the insertion to viewpoint $p_i$, which is traversed at that point by the sweep line. Analogously, the number of deletions is $O(m)$ as well.


Each insertion or deletion operation in the event queue has cost $O(\log m)$, since the queue only contains intersection events about rays that are  consecutive along the sweep line, and there can be at most $m$ rays intersected by the sweep line at a given time.
Since the total number of events processed is $O(m)$, it follows that the total time spent on maintaining the lower envelope of $L'$ is $O(m \log m)$.

Note that, even though we have presented this sweep line algorithm separately,
it should be interleaved with the main sweep line algorithm described in the previous
paragraphs.

\textbf{Maintaining the secondary viewpoint.}
On top of the updates caused by cases (i.2), (ii.3) and (ii.4), which take $O(1)$ time, we do the following: Every time that there is a new secondary viewpoint or a new lowermost ray in $L'$, we check whether $r_b$ intersects this lowermost ray $r_j$. In the affirmative, we keep the $x$-coordinate of the intersection point and, every time that a new iteration starts, we check whether this $x$-coordinate lies between the $x$-coordinates of the two extremes of the edge. If these happens and at that point the secondary viewpoint and the lowermost ray in $L'$  have not changed, then $p_j$ becomes the new secondary viewpoint. Thus, $r_j$ is removed from $L'$, and the lowermost ray in $L'$ is updated. Notice that the ray
corresponding to the old secondary viewpoint is not added to $L'$ because it
corresponds to a viewpoint that is now dominated.

As mentioned earlier, these operations are performed every time that there is a new secondary viewpoint or a new lowermost ray in $L'$. If there is a new secondary viewpoint $p_b$ caused by event (ii.3), we can associate it to the vertex $v$ such that $r_b=\rho(p_b, v)$. If there is a new secondary viewpoint because $r_b$ intersects the lowermost ray $r_j$ in $L'$, we can associate it to the old secondary viewpoint, which becomes dominated and thus disappears from the algorithm. This shows that the secondary viewpoint changes $O(n+m)$ times throughout the whole algorithm. On the other hand, we already know that the lowermost ray in $L'$ changes at most $O(m)$ times. 
Thus, the operations described in the paragraph above are globally done in $O(n+m)$ time.

\paragraph{Correctness and running time.} 

The correctness of the method follows from the fact that all changes in the terrain between visible and invisible are detected. Observation~\ref{obs:left-vis} guarantees that it is enough
to keep track of only the leftmost visible viewpoint, which we use in (i.4) and (ii.5). Finally, Corollary~\ref{cor:crossing} shows that, whenever two rays in $L$ cross, one of them stops being relevant for the algorithm. We use this property in the definition of $L$.

Next we analyze the running time. As seen before, maintaining the lowermost ray in $L'$ at any time can be done in $O(m\log m)$ time. 
Notice that the number of insertions to $L$ can be as high as $\Theta(n)$, so if we maintained the lowermost ray of $L$ instead of that of $L'$, the running time would increase to $O(n\log m)$. 
This is the reason for keeping the secondary viewpoint separate from the remaining rays in $L$, and maintaining $L'$ instead of $L$. 

Other than that, we spend constant time per iteration. Recall that the number of iterations is bounded by the sum of:
(i) the number of vertices in $\T$,
(ii) the number of times that there is a new non-empty secondary viewpoint,
(iii) the number of times that there is a new non-empty lowermost ray of $L'$,
(iv) the number of times that we are in event (i.2) and $r_b$ intersects 
  $wv$,
(v) the number of times that we are in event (i.4) and the lowermost
 ray $r_j$ 
 of $L'$ intersects $wv$. We have already argued that the sum of (ii) and (iii) is $O(n+m)$. Furthermore, (v) can be seen as a particular case of (iii). Regarding (iv), this case leads to a change in the status of the terrain, from invisible to visible, which happens $O(n)$ times during the whole algorithm. Consequently, we have that the total number of iterations is $O(n+m)$. 



To conclude, we observe that the
right-visibility map can be computed analogously. We finally merge
the two maps in $O(n)$ time and obtain the visibility map. Note that the algorithm can also be modified to
output, for each visible region, a set of viewpoints that cover that region. 

We obtain the following theorem:

\begin{theorem} \label{thm:1.5-map-alg}
Given a 1.5D terrain \T, the visibility map  $\vis ({\cal T},\G)$ can be constructed in $O(n+m \log m)$ time.
\end{theorem}

\subsubsection{Computing the Colored Visibility Map}
\label{subsec:1.5-alg-col}

In this subsection, we assume that no three vertices of $\T$ are aligned. This assumption is made to avoid certain configurations of vertices and viewpoints that can generate visibility regions where a great number of viewpoints become visible, but that consist of a single point. Such situations complicate the presentation of the algorithm and do not have practical utility.

The computation of the colored visibility map is similar to that
of $\vis ({\cal T},\G)$, with the extra requirement of having to
maintain all visible viewpoints during the sweep. We first compute the left- and right- colored visibility maps, and then we merge them. We now explain the computation of the left-colored visibility map (thus,
\emph{visible} stands for \emph{left-visible}).

As in the algorithm for $\vis ({\cal T},\G)$, we sweep the terrain from left to right.
The sweep itself, however, becomes much simpler, since it must stop every time a viewpoint changes its visibility status.
The state of the sweep line is maintained with a balanced binary tree containing all currently visible viewpoints, sorted by $x$-coordinate.

The events of the sweep, kept in a priority queue, can be of two types: $x$-coordinates of terrain vertices (\emph{vertex events}), and $x$-coordinates of points of the terrain where some viewpoint becomes visible (\emph{viewpoint events}). In particular, we do not maintain a set of shadow rays as in the algorithm for $\vis ({\cal T},\G)$; instead, every time that a viewpoint stops being visible, we compute the first point (if any) of the terrain where it becomes visible again, and we add an event at that point. The details are explained below.

By the assumption that no three vertices of $\T$ are aligned, when a vertex event and a viewpoint event have the same $x$-coordinate, then the viewpoint event is caused by a viewpoint lying precisely in that position (in other words, no viewpoint can become visible again at another vertex). In this case, where two events have the same $x$-coordinate, vertex events are
processed first. 
Viewpoint events are handled by updating
the list of visible viewpoints. 
In vertex events we compute all
viewpoints that become invisible after the vertex, and identify
where they reappear. 
The latter is computed by a ray-shooting query with the corresponding shadow ray of each disappearing viewpoint.
Observation~\ref{obs:left-vis} implies that all the viewpoints that become invisible at a given vertex are contiguous in the list of visible
viewpoints (i.e. if a viewpoint becomes invisible, then all viewpoints to its right also do),
hence the leftmost visible viewpoint that disappears at a given vertex can be located in $O(\log m)$ time.
Moreover, the rightmost visible viewpoint is always included, thus the interval of visible viewpoints that become invisible at a given vertex can be determined in $O(\log m)$ time. 

For each of these viewpoints $p_i$, we spend $O(\log m)$ time to remove $p_i$ from the binary tree containing the visible viewpoints. Additionally, we perform a ray-shooting query to detect the point where $p_i$ reappears. By preprocessing the terrain (seen as a simple polygon) in $O(n)$ time and space, ray-shooting queries are answered in $O(\log n)$ time~\cite{cegghss-rspugt-94}. Finally, we add a viewpoint event at the point where $p_i$ reappears, which takes $O(\log n)$ time. So we can process all viewpoints that become invisible in
$O(\log m+\log n)=O(\log n)$ time for each of them.

The time spent on each viewpoint that disappears at a given vertex is charged to the point where the viewpoint becomes visible again, which is the starting point of a new region of $\colvis ({\cal T},\G)$.
Note that it can happen that several viewpoints reappear at exactly the same point and thus the computations done for these viewpoints are all charged to the starting point of the same region of $\colvis ({\cal T},\G)$.
To upper-bound the total number of points where viewpoints reappear simultaneously, we use the following lemma:

\begin{lemma} \label{lem:simul-vis}
For every pair of viewpoints $p_i$ and $p_j$, there exists at most one point $q \in \T$ 
such that $p_i$ and $p_j$ change from invisible to visible at $q$.
\end{lemma}

\begin{proof}
Let us suppose that $p_i$ is to the left of $p_j$. Let $q$ be a point in $\T$ such that $p_i$ and $p_j$ change from invisible to visible at $q$. 
Clearly, $x(p_i)< x(q)$. Therefore, there 
exists a vertex $v_k\in \T$ such that $\rho(p_i,v_k)$ is a shadow ray intersecting $\T$ at $q$.
We notice that $x(v_k)\leq x(p_j)$: Otherwise, $p_j$ would be below the segment $p_iv_k$ and $v_k$ would block the visibility between $p_j$ and $q$ (see Figure~\ref{fig:reappearing-sim}, left). Since we are assuming that no three vertices of $\T$ are aligned, we have that $p_j\neq v_k$ and $p_j\neq q$, so $x(v_k)<x(p_j)< x(q)$. We conclude that $\T(v_k, p_j)$ is below $\rho(p_i,v_k)$ (see Figure~\ref{fig:reappearing-sim}, right).

If there existed another point $q'\in \T$ such that $p_i$ and $p_j$ change from invisible to visible at $q'$, then there would exist another vertex $v_l\in \T$ ($l\neq k$) such that $\T(v_l, p_j)$ is below $\rho(p_i,v_l)$.
This would imply that at $p_j$ there would be two shadow rays corresponding to $p_i$, contradicting Observation~\ref{obs:rays}.
\end{proof}

\begin{figure}[tb]
\centering
\includegraphics{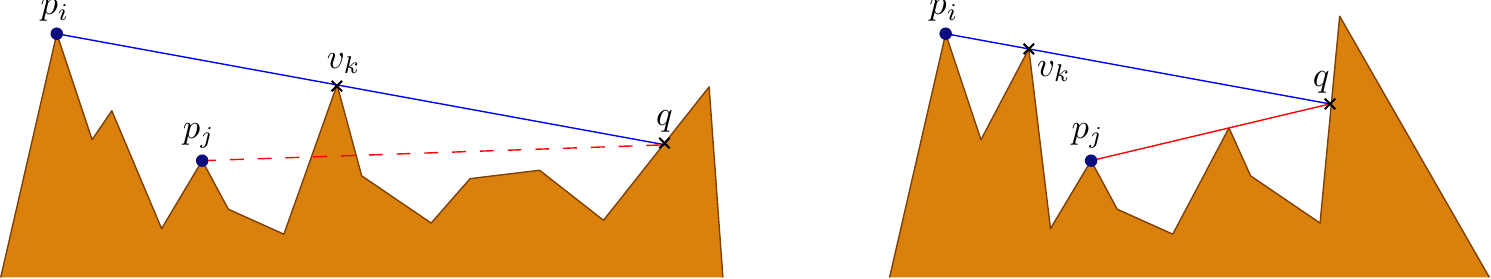}
\caption{Proof of Lemma~\ref{lem:simul-vis}. Left: If $x(v_k)> x(p_j)$, then $p_j$ does not see $q$. Right: $\T(v_k, p_j)$ is below $\rho(p_i,v_k)$.}
\label{fig:reappearing-sim}
\end{figure}

In consequence, a pair of
viewpoints cannot reappear simultaneously more than once. 
Thus the total number of points of $\T$ where more than one viewpoint reappear simultaneously is $O(m^2)$. This leads to the following
result. 

\begin{theorem} \label{thm:1.5-colvis-alg}
Given a 1.5D terrain \T, the colored visibility map $\colvis ({\cal T},\G)$ can be
constructed in $O(n + (m^2 + k_c )\log n)$ time, where $k_c$ is the complexity of $\colvis(\T,\S)$.\footnote{With a
finer analysis, the $O(m^2)$ term can be reduced to $O(m\sqrt{\eta})$, for
$\eta$ the number of points of $\T$ where more than one
viewpoint become visible ($\eta\leq \binom{m}{2}$).}
\end{theorem}

Note that the $O(m^2)$ term in the running time disappears if one assumes that at no point of the terrain more than $O(1)$ viewpoints change between visible and invisible.
Moreover, notice that the algorithm can be modified to output $\colvis (\T, \S)$ with the following extra information: for the first component of the map, the algorithm
returns the set of viewpoints that see that component. For the
subsequent components, the algorithm outputs all changes in the
set of visible viewpoints with respect to the component right before. The running time remains the same and becomes optimal up to a $O(\log n)$ factor.

\subsubsection{Computing the Voronoi Visibility Map}

\paragraph{A Simple Divide and Conquer Approach.}  A simple way to construct the Voronoi visibility map $\vorvis(\T,\S)$  is using divide and conquer. First, divide set $\S$ into
two sets $\S_1$ and $\S_2$, each with $\frac{m}{2}$ viewpoints.
Each set is recursively divided into two until each subset has
only one viewpoint. The Voronoi diagram of one viewpoint is its
visibility polygon, which can be computed in $O(n)$
time~\cite{js-clvpa-87}. Each of these diagrams can be transformed
into a list of intervals such that each interval defines a portion of the
terrain that is assigned to a particular viewpoint (or none).
Therefore, the merge of two smaller diagrams can be done by
comparing both intervals where there are parts of the terrain
visible from the two viewpoints, and choosing the closest one. In
other words, that portion of the terrain is intersected with
the perpendicular bisector between both viewpoints. In general,
the merge step of two diagrams at the $i$th level (which have $m/2^i$ viewpoints each) takes $O((m/2^i)n)$ time, which adds up to $O(mn)$ time over all pairs-to-be-merged at the same level.
Thus the whole procedure takes $O(mn \log m)$ time.


\paragraph{An Output-Sensitive Algorithm.} \label{sec:alg1.5Vor} Even though $\vorvis(\T,\S)$ can have $\Theta(mn)$ complexity, it
seems unlikely that such high complexity arises often in practical
applications. In the following we present an alternative algorithm
that essentially extracts the Voronoi visibility map from the
colored visibility map. Its running time depends on the complexity
of the two structures, and avoids the fixed $O(mn)$ term of the
previous method. Since we use the algorithm for $\colvis(\T,\S)$, here we also assume that no three vertices of $\T$ are aligned.

Our algorithm starts by computing $\colvis(\T,\S)$ using the algorithm from Section~\ref{subsec:1.5-alg-col}. We then sweep the terrain from left to right. During this
sweep, we maintain three data structures: 
\begin{enumerate}
\item  A doubly-linked list with the vertices of $\colvis(\T,\S)$, sorted from left to right.
\item A list $\S'$ with the currently visible viewpoints.
\item For each currently visible point, we keep track of when it became visible for the last time. More precisely, for each  $p_i\in \S'$ we keep a point $a_i$ that is the starting point  of the last visible region of $p_i$ encountered so far.
\end{enumerate}

Recall that we use $\T[a,c]$, for $a,c$ on $\T$ and $x(a) < x(c)$, to
denote the closed portion of the terrain between $a$ and $c$. The
algorithm produces $\vorvis(\T,\S)$ as a list of interval-viewpoint pairs $([a,c],p_i)$, such that $p_i$ is the closest visible
viewpoint to all points in $\T[a,c]$. If $\T[a,c]$ is not visible
from any viewpoint, $p_i$ is set to $\bot$.

Our algorithm uses the following two functions, whose
implementation is described later.

\textsc{IsAlwaysCloser}$([a,c],p_1,p_2)$ determines whether $p_1$
is always closer than $p_2$ in $\T[a,c]$, assuming both viewpoints
are visible throughout $\T[a,c]$.

\textsc{FirstRegionChange}$([a,c],p_1,\S')$ assumes that $p_1$ is
visible throughout $\T[a,c]$ and is the closest visible viewpoint
at $a$; it returns the leftmost point in $\T[a,c]$ where $p_1$
stops being the closest visible viewpoint from $\S'$ (or the end
of the interval, if that never happens).

We process \T in a number of iterations. Each iteration
starts at the leftmost point $u$ of a new Voronoi region, with $\S'$ containing
the viewpoints that are visible from $u$.

If $\S'=\emptyset$, then the region starting at $u$ and ending at the start
point $v$ of the next region in $\colvis(\T,\S)$ is not visible from any
viewpoint. We report the region $[u,v]$ with $\bot$, and move forward (towards
the right) until $v$, where a new Voronoi region, and thus a new iteration,
starts.

If $\S' \neq\emptyset$, we compute the closest visible viewpoint in $O(m)$
time; if there is more than one, we move infinitesimally to the right of $u$,
and compute the closest visible viewpoint there. Without loss of generality, we
assume that the closest visible viewpoint is $p_1$. For all viewpoints $p_i\in
\S'$, we set $a_i := u$. We now start traversing the terrain, from $u$ towards
the right. At a point $q$, we might find several events from $\colvis(\T,\S)$:

\begin{enumerate}
\item A viewpoint $p_j$ becomes visible.
We update $\S'$, set $a_j:=q$, and continue the sweep.

\item A viewpoint $p_j \neq p_1$ becomes invisible.
  We update $\S'$ and proceed depending on two subcases:

  \begin{enumerate}
  \item \textsc{IsAlwaysCloser}$([a_j,q],p_1,p_j)$ = \textsc{True}.
    Continue traversing the terrain.

  \item \textsc{IsAlwaysCloser}$([a_j,q],p_1,p_j)$ = \textsc{False}.  There is
    a point in $\T[a_j,q]$ at which $p_j$ is closer than $p_1$, so at least one
    Voronoi region starts between $u$ and $q$. We find the leftmost region
    change $v$ by calling \textsc{FirstRegionChange}$([u,q],p_1,\S')$,
 and report $[u,v]$ as a Voronoi region
    with $p_1$ as closest point. We now have to start a new iteration of the algorithm from $v$. In order to do that, first we have to restore $\S'$ to the set of visible viewpoints at $v$. To that end, we backtrack our sweep, i.e. we traverse  the sequence of $\colvis(\T,\S)$ events from right to left (updating $\S'$ as we encounter events) until we reach $v$.
  \end{enumerate}

\item Viewpoint $p_1$ becomes invisible. We update $\S'$, and compute
  \textsc{IsAlwaysCloser} $([a_i,q],p_1,p_i)$, for all $p_i \in \S'$. If the
  answer is \textsc{True} for \emph{all} viewpoints in $\S'$, we report the
  region $[u,q]$ with $p_1$ as closest viewpoint, and start a new Voronoi
  region and a new iteration at $q$.
Otherwise, there is at least one Voronoi region that starts between $u$
  and $q$. We handle this analogously to case 2(b).
\item Vertex $q$ is the last (rightmost) vertex of \T. We proceed as in case 3, except that no new Voronoi region starts after $q$.
\end{enumerate}

After processing all events, we have successfully computed $\vorvis(\T,\S)$.

Since we backtrack our sweep in step 2(b), it may be the case that we
(unnecessarily) visit events from $\colvis(\T,\S)$ multiple times.
We can avoid this by augmenting this step as follows. Consider
step 2(a). We notice that there cannot be a Voronoi region of $p_j$
between $a_j$ and $q$ (since at least $p_1$ is closer and
visible). So we can remove the events of $p_j$ becoming visible at
$a_j$ and invisible at $q$ from the list of $\colvis ({\T},\S)$ events.
We remove the event at $q$ in step 2(a) itself.
The event at $a_j$ is removed if we encounter it while
backtracking in step 2(b): at each event of type 1, i.e. a viewpoint
$p_j$ becoming visible, we check if $p_j$ is in $\S'$. If not, we
must have removed the event at $q$ (where $p_j$ becomes invisible) from
$\colvis(\T,\S)$. Thus we can also remove the event at $a_j$.

Next, we analyze the running time of the above algorithm.
The running times needed to implement \textsc{IsAlwaysCloser} and  \textsc{FirstRegionChange} will be analyzed later. For the time being, we refer to them in an abstract way.
 Let $A(n,m)$ be the time
required for a (single) \textsc{IsAlwaysCloser} query, and let $F'(n,m,v) =
F(n,m) + v$ be the time required for a \textsc{FirstRegionChange} query, where $v$ is the number of terrain vertices between $u$ and the point returned (later it will become clear why $F'$ depends also on $v$). 
We now show that:

\begin{lemma}
  \label{lem:time_spent_on_function_calls}
  The total time that the Voronoi visibility map algorithm spends on
  \textsc{IsAlwaysCloser} and \textsc{FirstRegionChange} queries is $O((k_c + m^2+
  mk_v)A(n,m) + k_vF(n,m) + n)$.
\end{lemma}

\begin{proof}
  We call \textsc{IsAlwaysCloser} or \textsc{FirstRegionChange} only in events
  of type 2 and 3. At an event of type 2(a), we remove a vertex from the
  colored visibility map, so there are at most $k_c+ m^2$ such events. At events of
  types 2(b) and 3, we report a new Voronoi region. Since we report each region
  at most once, the number of events in these cases is at most $k_v$.

  It now follows that the number of calls to \textsc{IsAlwaysCloser} and
  \textsc{FirstRegionChange} is at most $O(k_c + m^2+ mk_v)$ and $O(k_v)$,
  respectively.

  A \textsc{FirstRegionChange} query takes $F'(n,m,v) = F(n,m) + v$ time, where
  $v$ is the number of terrain vertices in the interval between $u$ and the point returned. Each
  such interval corresponds with a Voronoi region, and each terrain vertex
  occurs in at most two Voronoi regions, so all values $v$ sum up to $O(n)$.
\end{proof}

\begin{lemma}
  \label{lem:1.5D-vor-alg}
  Given a 1.5D terrain \T and its colored visibility map $\colvis(\T,\S)$,
  the Voronoi visibility map $\vorvis({\T},\S)$ can be constructed in $O((k_c + m^2+ mk_v)A(n,m) +
  k_vF(n,m))$ time.
\end{lemma}

\begin{proof}
  The correctness of the algorithm is clear from its description. So what
  remains is to analyze its running time.

  Each time we visit a vertex of (our copy of) $\colvis(\T,\S)$ we do a
  constant amount of work, and possibly do some \textsc{IsAlwaysCloser} and/or   \textsc{FirstRegionChange} queries. By
  Lemma~\ref{lem:time_spent_on_function_calls}, the total time spent on these
  queries is $O((k_c + m^2+ mk_v)A(n,m) + k_vF(n,m) +n)$. We now show that the total
  number of vertices of $\colvis(\T,\S)$ visited is at most $O(k_c + m^2+ mk_v)$. The lemma then
  follows.


We start by counting the number of vertices of our copy of $\colvis(\T,\S)$
  that we visit while sweeping backwards (that is, from right to left). During
  a single backward sweep from a point $q$, we might encounter events that correspond to viewpoints belonging to $\S'$ at point $q$ becoming visible (that is, becoming visible when sweeping from left to right). For each viewpoint that belongs to $\S'$ at $q$, we charge the first event of this type to the Voronoi region starting at $q$. Since we
  sweep backwards at most $k_v$ times, namely only when we found a new Voronoi
  region, the number of such events is at most $O(mk_v)$. If we encounter an
  event for a viewpoint not in $\S'$ at $q$, or the second (or third, etc.) event for a viewpoint in $\S'$ at $q$, we immediately delete it from (our copy
  of) $\colvis(\T,\S)$.\footnote{Note that we only encounter points where a
    visibility region starts (i.e.~events of type 1), since we already removed
    their corresponding end points in step 2(a).} Hence, this happens at most
  $k_c+ m^2$ times in total. Thus, the number of vertices that we visit while
  sweeping backwards is $O(k_c+ m^2 + mk_v)$.

  Our sweepline moves continuously, starting at the first vertex of the colored
  visibility map, and ending at the last vertex. Thus, the number of times we
  visit a vertex while sweeping forwards (left to right) is one plus the number
  of times we visit the vertex while sweeping backwards. It follows that the
  total number of times we visit a vertex while sweeping forwards is also at
  most $O(k_c + m^2+ mk_v)$.
\end{proof}

\paragraph{Implementing \textsc{IsAlwaysCloser} and
  \textsc{FirstRegionChange}.} It is fairly straightforward to implement
\textsc{IsAlwaysCloser} to run in $O(\log n)$ time: we simply use a
ray-shooting query, where the ray is the bisector of the two viewpoints
involved. However, as we show next, we can even answer this question in
constant time.

\begin{lemma} \label{lem:bisInt}
Consider two points $r$ and $t$ such that all of $\T[r,t]$ is
visible from two viewpoints $p_1$ and $p_2$. We can decide whether
there exists some point in $\T[r,t]$ that is closer to $p_2$ than
to $p_1$ in $O(1)$ time.
\end{lemma}

\begin{proof}
  We claim that in case that neither $r$ or $t$ lies on the bisector $b_{1,2}$
  of $p_1$ and $p_2$ then there is a point $q \in \T[r,t]$ that is closer to
  $p_2$ than to $p_1$ if and only if one of the following conditions hold:
  \medskip
  \begin{itemize}[nosep]
  \item[(i)] $r$ or $t$ is closer to $p_2$ than to $p_1$;
  \item[(ii)] $r$ and $t$ are closer to $p_1$ than to $p_2$, and $x(r)<x(p_2)<x(t)$.
  \end{itemize}
  \medskip
  We can check these two conditions in constant time. In case $r$ lies on the
  bisector we use the above characterization for a point $r'$ slightly to the
  right of $r$. In case $t$ lies on the bisector we do the same for a point
  $t'$ slightly to the left of $t$.



  Let us now prove that this characterization is correct. Notice that it
  suffices to prove that the characterization is correct for the
  case where neither $r$ or $t$ lies on the bisector $b_{1,2}$.

  Clearly, if condition (i) is satisfied, there exists some point in
  $\T[r,t]$ closer to $p_2$ than to $p_1$. If condition (ii) is
  satisfied, then $p_2$ belongs to $\T[r,t]$, and obviously the
  point of ${\cal T}$ where $p_2$ lies is closer to $p_2$ than to
  $p_1$.

  It remains to prove that, if there exists some point in $\T[r,t]$
  closer to $p_2$ than to $p_1$, then either (i) or (ii) holds.
  Assume, for the sake of contradiction, that this is not true. Then
  there exists some point $q$ in $\T[r,t]$ closer to $p_2$ than to
  $p_1$, $r$ and $t$ are closer to $p_1$ than to $p_2$, and either
  $x(p_2)< x(r)$ or $x(p_2)> x(t)$. 
Since $r$ and $t$ are closer to $p_1$ than to $p_2$, the bisector $b_{1,2}$ is a line that properly
  intersects\footnote{By a \emph{proper}
    intersection of $b_{1,2}$ and the terrain we mean an intersection
    such that right before the intersection point the terrain is closer to
    one of the viewpoints, and right after it is closer to the other
    viewpoint.} $\T[r,t]$ an even number of times (at least two), so in
  particular $b_{1,2}$ is not vertical.

  Let us first suppose that $p_2$ is below $b_{1,2}$. Then $r$ and
  $t$ are above $b_{1,2}$.

  If $x(p_2)< x(r)$, then we have the
  situation in Figure~\ref{fig:lemmachangevoronoi} (left). Since $q$ is
  visible from $p_2$, it lies
  above the line through $p_2$ and $r$. Since it is closer to $p_2$ than to $p_1$, it lies
  below $b_{1,2}$. But these two half-planes do not intersect to the right of
  $r$. This yields a contradiction. The case $x(p_2) > x(t)$ is symmetrical.

  \begin{figure}[tb]
    \centering
    \includegraphics{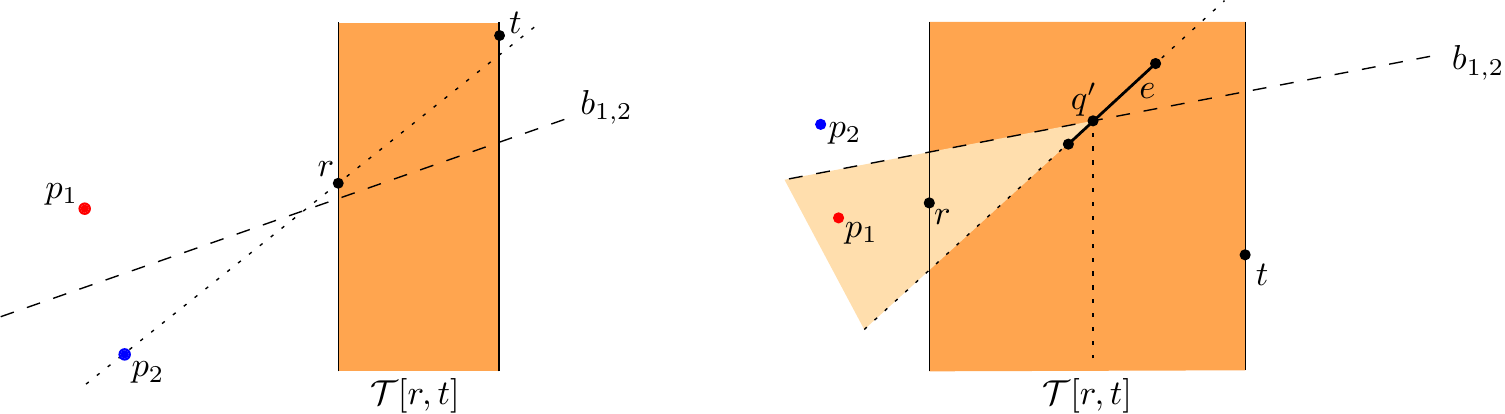}
    \caption{The two cases in the proof of Lemma~\ref{lem:bisInt}. Left: $p_2$ is below $b_{1,2}$ and $x(p_2)< x(r)$. The shaded area indicates the region where $\T[r,t]$, and thus point $q$, lies. Right: $p_2$ is above $b_{1,2}$.
      The lighter shaded area indicates the region where $p_1$ lies, while the darker shade denotes the region where $\T[r,t]$ lies.}
    \label{fig:lemmachangevoronoi}
  \end{figure}

	Next we suppose that $p_2$ is above $b_{1,2}$. 
  Then $r$ and $t$ are below $b_{1,2}$. 
  Let $e$ be the first edge encountered when traversing the terrain from $r$ to $t$ such that a portion of $e$ of positive length is closer to $p_2$ than to $p_1$, and let $q'$ be the first point at which that happens (i.e., the point at which $b_{1,2}$ intersects \T). See Figure~\ref{fig:lemmachangevoronoi} (right).
  Since $p_1$ sees $e$,
  $p_1$ lies above the line through $e$.
  On the other hand, $p_1$ lies below
  $b_{1,2}$.
  In particular, $p_1$ lies to the left of $q'$, while
  $t$ lies to its right. Therefore the segment $p_1t$ intersects
  the vertical ray with origin at $q$ and going downwards, so $t$
  is not visible from $p_1$. This yields a contradiction.
\end{proof}

We can easily implement \textsc{FirstRegionChange}$([u,q],p_1,\S')$ to run in
$O(m\log n)$ time. For every $p_i \in \S'$, we use a ray-shooting query to
compute the leftmost point (if any) on $\T[a_i,q]$ that is closer to $p_i$ than
to $p_1$. We then simply select the leftmost point $u'$ among the
results. Next, we show that we can improve this to $O(m+\log n \log m + n')$
time, where $n'$ is the number of vertices of \T between $u$ and the point
where there is the first change of Voronoi region.

Recall that a change in the Voronoi region can take place due to two situations: a viewpoint closer than $p_1$ reappearing, or the terrain crossing the bisector between $p_1$ and another visible viewpoint. We treat these two situations separately in two phases.
In the first phase we filter the viewpoints in $\S'$ in order to leave
precisely the ones that can create a region change only because of a bisector
crossing $\T[u,q]$. During the filtering, we also find the leftmost point
  (if any) where a viewpoint closer than $p_1$ reappears. This is a candidate point to produce a change in the Voronoi region.
In a second phase we compute the leftmost point (if any) where a bisector between $p_1$ and one of the remaining viewpoints crosses $\T[u,q]$.
The computed point is the second candidate point to produce a change in the Voronoi region.
Next, we present these two phases in detail.

For the remaining of this section, we sometimes denote $\T[a,c]$ by simply $[a,c]$. Furthermore, for each $p_i \in \S'$, the ``interval" $[a_i, q]$ (that is, the portion of the terrain $\T[a_i,q]$) is denoted by $I_i$. 
We assume that no viewpoint in $\S'$ has a vertical bisector with $p_1$ (this special, but simpler, case will be discussed later).

\paragraph{Filtering phase.}

In a first step, we remove from $\S'$ all the viewpoints $p_i\in
\S'$, such that no point in $I_i$ is closer to $p_i$ than to
$p_1$. Using Lemma~\ref{lem:bisInt}, this can be done in $O(m)$
time.

In a second step, we check for every $p_i\in \S'$ whether $a_i$ is
closer to $p_i$ than to $p_1$. From all the viewpoints such that
the answer is affirmative, we only keep the viewpoint $p_k$ such
that $a_k$ is leftmost; all the other viewpoints for which the
answer is affirmative are removed from $\S'$. Viewpoint $p_k$ is
also removed from $\S'$. 

Note that the region of $p_1$ cannot extend further to the right than $a_k$. Therefore, $a_k$ is a candidate point to produce a change in the Voronoi region. 
However, it can happen that there is no $a_k$ because no viewpoint satisfies the required condition; in this case, $a_k$ in the following explanation must be replaced by $q$.

Next, all
viewpoints such that $x(a_i)>x(a_k)$ are removed from $\S'$. All
viewpoints such that no point in $\T[a_i,a_k]$ is closer to $p_i$ than to
$p_1$ are also removed from $\S'$. We denote by $\mathcal{R}$ the set of remaining points in $\S'$ (possibly $\mathcal{R}=\emptyset$, in which case we are done).


Any remaining viewpoint $p_i \in \mathcal{R}$ has the following properties:
\begin{enumerate}
\item $p_1$ is closer to $a_i$ than $p_i$. \vspace{-6pt}
\item $p_i$ is closer to $p_1$ in at least one point in $\T[u,a_k]$. \vspace{-6pt}
\item $x(a_i) \leq x(a_k)$.
\end{enumerate}

Property 1 implies that if there is a region change between $u$ and $a_k$, it must be due to the terrain crossing a bisector (i.e. it cannot be due to a viewpoint reappearing).
Therefore in the following phase we find the leftmost point where $\T[u,a_k]$ crosses a bisector between $p_1$ and a viewpoint in $\mathcal{R}$.

\paragraph{Computing the first terrain-bisector crossing.}
We show next that the first terrain-bisector crossing can be computed efficiently by using a prune and search algorithm.

\begin{lemma} \label{lem:PruneAndSearch}
  The leftmost point where $\T[u,a_k]$ crosses a bisector between $p_1$ and a
  viewpoint in $\mathcal{R}$ can be found in $O(m+\log n \log m +n')$ time,
  where $n'$ is the number of vertices of \T between $u$ and the crossing point
  found.
\end{lemma}

\begin{proof}
We use a prune and search algorithm inspired by Megiddo's technique to solve linear programming in $\mathbb{R}^2$ in linear time~\cite{Megiddo}. At every iteration of the algorithm, roughly one fourth of the viewpoints in $\mathcal{R}$ are dropped because it is inferred that they do not cause the leftmost intersection between a bisector and $\T[u,a_k]$. We next give the details of the method.

Recall that we assume no bisector is vertical.
We divide the remaining viewpoints in $\mathcal{R}$ into two groups.
We add to set $\mathcal{R}_1$ all viewpoints $p_i\in \mathcal{R}$ such that $p_1$
is below $b_{1,i}$. Then we add to set $\mathcal{R}_2$ all viewpoints
$p_i\in \mathcal{R}$ such that $p_1$ is above $b_{1,i}$.\footnote{Note that this is equivalent to saying that $\mathcal{R}_1$ contains all viewpoints that are higher than $p_1$, and $\mathcal{R}_2$ contains all viewpoints that are lower than $p_1$.}
 We define $u''$
as the leftmost point in $\T[u,a_k]$ such that at $u''$
there is a change of region in the Voronoi visibility map, and the
viewpoint responsible for the change is in $\mathcal{R}$.
While there is more than one viewpoint in $\mathcal{R}_1$ or $\mathcal{R}_2$, we
repeat the following procedure.

We take the viewpoints in $\mathcal{R}_1$ and put them in pairs
arbitrarily. We do the same with the viewpoints in $\mathcal{R}_2$. We
consider all pairs of viewpoints
$p_i,p_h\in \mathcal{R}_1$ or $p_i,p_h\in \mathcal{R}_2$ such that the
$x$-coordinate of the intersection point of $b_{1,i}$ and
$b_{1,h}$ lies in the interval $[x(a_{i,h}),x(c_{i,h})]$, where
$[a_{i,h},c_{i,h}]=I_i\cap I_h$. We compute in $O(m)$ time the median
$\mu$ of the $x$-coordinates of these intersection points. 
In $O(\log n)$ time, we compute the corresponding point on the terrain, $\T[\mu]$. Using Lemma~\ref{lem:bisInt}, we decide
in $O(m)$ time whether $x(u'')\leq \mu$ or $x(u'')> \mu$: We have
that $x(u'')\leq \mu$ if and only if at least one of the
viewpoints $p_i\in \mathcal{R}$ satisfies that there exists
some point in $I_i\cap [u,\T[\mu]]$ closer to $p_i$ than to $p_1$.

We next show that in both cases ($x(u'')\leq \mu$ and $x(u'')> \mu$) we can discard enough viewpoints from $\mathcal{R}$.

Let us first suppose that $x(u'')\leq \mu$. We start by selecting all pairs $p_i,p_h\in
\mathcal{R}_1$ such that $b_{1,i}$ and $b_{1,h}$ do not
intersect, or the $x$-coordinate of the intersection point of
$b_{1,i}$ and $b_{1,h}$ does not lie in
$[x(a_{i,h}),x(c_{i,h})]\cap [x(u),\mu]$. We will prove that we can remove $p_i$ or $p_h$ (or both) from
$\mathcal{R}_1$. 
 If for one of the two viewpoints, say $p_i$, it holds
that no point in $I_i\cap [u,\T[\mu]]$ is closer to $p_i$ than to
$p_1$, then we remove that viewpoint from $\mathcal{R}_1$ (since we know that $x(u'')\leq \mu$). If this holds
for both viewpoints, we remove both of them from $\mathcal{R}_1$.
Otherwise, we assume without loss of generality that $x(a_i)\leq
x(a_h)$. Recall that, for all
$p_j\in \mathcal{R}_1$, $a_j$ is closer to $p_1$ than to $p_j$. By the
definition of $\mathcal{R}_1$, this implies that $a_j$ is below $b_{1,j}$.

Now there are three possible situations, depending on the positions of $a_h$ and $b_{1,h} \cap \{x=x(a_h)\}$ with respect to $b_{1,i}$
(see Figure~\ref{fig:lemPrunSear-2}).
In the first situation, depicted in Figure~\ref{fig:lemPrunSear-2}
(left), $b_{1,h}$ is below $b_{1,i}$ in the interval
$[x(a_{i,h}),x(c_{i,h})]\cap [x(u),\mu]$. We first check whether
there exists some point in the portion of the terrain between
$a_i$ and $a_h$ closer to $p_i$ than to $p_1$. In the affirmative,
we can remove $p_h$ from $\mathcal{R}_1$, since there is a change of
Voronoi region before $p_h$ becomes visible. Otherwise, the
portion of the terrain between $a_h$ and $\T[\mu]$ crosses both
$b_{1,i}$ and $b_{1,h}$, but notice that it crosses $b_{1,h}$
first. So we can remove $p_i$ from $\mathcal{R}_1$. In the second
situation, illustrated in Figure~\ref{fig:lemPrunSear-2} (center),
$b_{1,h}$ is above $b_{1,i}$ in the interval
$[x(a_{i,h}),x(c_{i,h})]\cap [x(u),\mu]$, and $a_h$ is above
$b_{1,i}$. In this case, the terrain crosses $b_{1,i}$ before
$p_h$ becomes visible, so we remove $p_h$ from $\mathcal{R}_1$. In the
third situation, depicted in Figure~\ref{fig:lemPrunSear-2}
(right), $b_{1,h}$ is above $b_{1,i}$ in the interval
$[x(a_{i,h}),x(c_{i,h})]\cap [x(u),\mu]$, and $a_h$ is below
$b_{1,i}$. We also remove $p_h$ from $\mathcal{R}_1$, since even though the
terrain crosses $b_{1,h}$ while $p_h$ is visible, it has crossed
 $b_{1,i}$ while $p_i$ is visible before.

\begin{figure}[tb]
    \centering
    \includegraphics{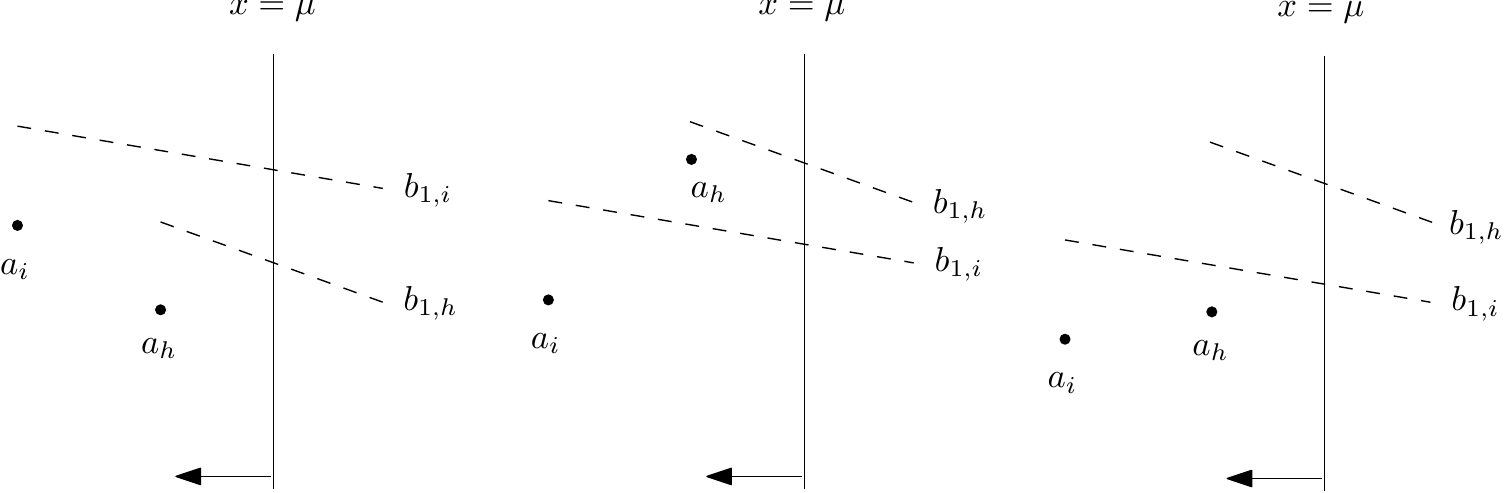}
    \caption{Three possible situations for the case $x(u'')\leq \mu$ and $p_i,p_h\in
\mathcal{R}_1$.} \label{fig:lemPrunSear-2}
\end{figure}

We do an analogous procedure with all pairs $p_i,p_h\in \mathcal{R}_2$
such that $b_{1,i}$ and $b_{1,h}$ do not intersect, or the
$x$-coordinate of the intersection point of $b_{1,i}$ and
$b_{1,h}$ does not lie in $[x(a_{i,h}),x(c_{i,h})]\cap
[x(u),\mu]$. We remove from $\mathcal{R}_2$ at least one of the
viewpoints of such a pair. 

By definition of $\mu$, we
have performed the above pruning operations to at least half of the pairs of
viewpoints in $\mathcal{R}_1$ and $\mathcal{R}_2$. Therefore, at least one fourth of the remaining viewpoints in $\mathcal{R}$ are discarded.

It remains to consider the situation where $x(u'')> \mu$. Here we apply a similar strategy: Let $p_i,p_h\in
\mathcal{R}_1$ be a pair such that $b_{1,i}$ and $b_{1,h}$ do not
intersect, or the $x$-coordinate of the intersection point of
$b_{1,i}$ and $b_{1,h}$ does not lie in
$[x(a_{i,h}),x(c_{i,h})]\cap [\mu,x(a_k)]$.

We consider essentially two cases.
In the first case, $x(a_i)\leq \mu$ and $x(a_h)\leq \mu$. We
assume without loss of generality that $b_{1,i}$ is below
$b_{1,h}$ in the interval $[x(a_{i,h}),x(c_{i,h})]\cap
[\mu,x(a_k)]$. Then we are in the situation illustrated in
Figure~\ref{fig:lemPrunSear-3}; in particular, $\T[\mu]$ lies below $b_{1,i}$ and $b_{1,h}$, since we know that $x(u'')> \mu$. The portion of the terrain between
$\T[\mu]$ and $a_k$ cross both $b_{1,i}$ and $b_{1,h}$, but it crosses
$b_{1,i}$ first. Thus, we remove $p_h$ from $\mathcal{R}_1$. 

\begin{figure}[!htb]
    \centering
    \includegraphics{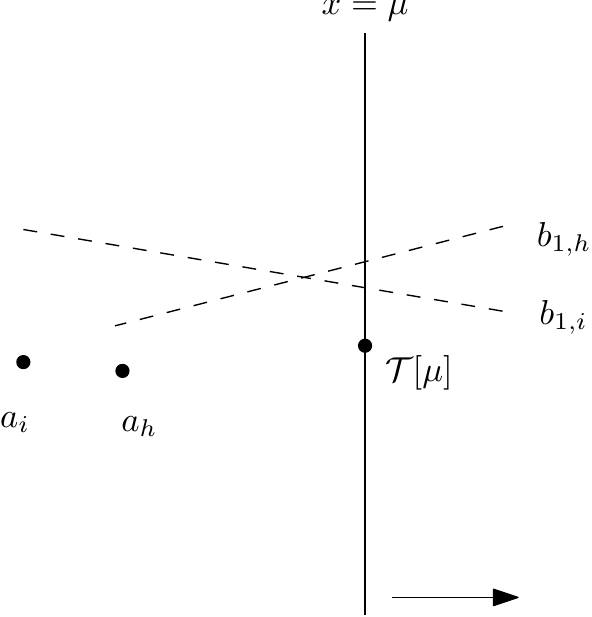}
    \caption{Case $x(u'')> \mu$ and $p_i,p_h\in
\mathcal{R}_1$.} \label{fig:lemPrunSear-3}
\end{figure}

In the second
case, $x(a_i)> \mu$ or $x(a_h)> \mu$. We assume without loss of
generality that $x(a_i)\leq x(a_h)$. This case can be further
subdivided into three cases, which essentially correspond to the
three cases depicted in Figure~\ref{fig:lemPrunSear-2}. These cases
are handled as explained above, and thus either $p_i$ or $p_h$ are
removed from $\mathcal{R}_1$.

By analogous arguments, we remove some viewpoints from $\mathcal{R}_2$ as
well. In total, we remove at least roughly one fourth of the
viewpoints in $\mathcal{R}_1 \cup \mathcal{R}_2$.

We repeat this procedure until there is one (or zero) viewpoint in
$\mathcal{R}_1$, and one (or zero) viewpoint in $\mathcal{R}_2$. The total running
time $T(n,m)$ of this part satisfies $T(n,m)\leq Cm+\log
n+T(n,3m/4)$ (for some constant $C$), which is $O(m+\log n \log
m)$.

We have now at most two viewpoints (one in $\mathcal{R}_1$, and one in $\mathcal{R}_2$) such that one of them is the viewpoint responsible for the leftmost change in the Voronoi
visibility map. In order to find it, we start traversing the
terrain from $u$ until we find the point $u'$ where the Voronoi
region starts. This takes time proportional to the size of the
terrain that we traverse.
\end{proof}

Therefore the first region change, assuming there are no vertical bisectors, occurs either at $a_k$ or at the point returned by the prune and search algorithm.

\paragraph{Vertical bisectors.}
It remains to explain how to deal with viewpoints in $\S'$ that have a vertical bisector with $p_1$.
First note that if there are several such viewpoints, then only the leftmost one is relevant.
This is because all viewpoints that have a vertical bisector with $p_1$ (and are relevant at this point of the algorithm) lie to the right of $p_1$ and have the same $y$-coordinate as $p_1$.
Thus the leftmost of them has its bisector crossing \T first and is the one that sees more of $\T[u,q]$.
Suppose now that there are viewpoints in $\S'$ with vertical bisectors with $p_1$, and let $p_v$ be the leftmost one.
Then the $x$-coordinate of the first region change is the minimum between the one obtained by the previous algorithm and $\max \{x(a_v), x(b_{1,v}) \}$,  where $x(b_{1,v})$ is such that the bisector $b_{1,v}$ has equation $x=x(b_{1,v})$.

In summary, we have proved the following result.

\begin{lemma}
  \label{lem:FirstbisInt}
  Let $[u,q]$ be an interval such that $p_1 \in \S$ is visible in
  all $\T[u,q]$ and is the closest visible viewpoint at $u$. Let
  $\S'$ be a set of viewpoints such that for each $p_i\in \S'$,
  $\T[a_i,q]$ is visible from $p_i$, for some $a_i$ such that $x(u)
  \leq x(a_i)$. Then in $O(m+\log n \log m +n')$ time we can find
  the leftmost point $u'\in \T[u,q]$ such that at $u'$ there is a
  change of region in $\vorvis({\T},\S')$, for  $n'$  the number of
  vertices in $\T[u,u']$.
\end{lemma}

We can compute the colored visibility map in $O(n + (m^2 + k_c )\log n)$ time
(Theorem~\ref{thm:1.5-colvis-alg}). By plugging in the results from
Lemma~\ref{lem:bisInt} and Lemma~\ref{lem:FirstbisInt} into
Lemma~\ref{lem:1.5D-vor-alg} we then obtain:

\begin{theorem}
  \label{thm:1.5D-vor-alg}
  Given a 1.5D terrain \T, the Voronoi visibility map $\vorvis({\T},\S)$ can be
  constructed in $O(n + (m^2 + k_c) \log n + k_v (m+\log n \log m))$ time.
\end{theorem}

 The algorithm assumes that no bisector of two viewpoints is
 collinear with an edge of $\cal T$, but it can easily be
 adapted to handle this case as well. On the other hand, if during the
 algorithm a relevant bisector is only tangent to the
 terrain (i.e., right before and right after this intersection the
 terrain is closer to the same viewpoint, among the two viewpoints
 defining the bisector), we do not consider this as a change in the Voronoi
 region, even if at
 the tangency point the terrain has two closest viewpoints.

 We also note that in the (hopefully unlikely) case of $k_v \in \Theta(mn)$,
 the simpler divide and conquer algorithm achieves a better running time than this algorithm.

\section{2.5D Terrains}

\subsection{Complexity of the Visibility Structures}

We start by showing that the visibility map of a 2.5D terrain can have $\Omega(m^2n^2)$ complexity.


\begin{proposition}
  \label{prop:lowerbound_vis_2.5d}
  The visibility map $\vis(\T,\G)$, for a 2.5D terrain \T, 
  can have complexity $\Omega(m^2n^2)$.
\end{proposition}

\begin{proof}
  We present a terrain that consists of a flat (horizontal) rectangle, the \emph{courtyard},
  surrounded by a thin wall, see Figure~\ref{fig:2.5D-VisMap-quad}. We make $O(n)$ (almost) vertical incisions, or
  \emph{windows}, in the northern and western walls. We place half of the
  viewpoints behind windows in the northern wall, and the other half behind
  windows in the western wall. Each viewpoint is placed so that it can see
  through $O(n)$ windows into the courtyard.
It follows that the visibility map inside the courtyard
  forms an $\Theta(mn) \times \Theta(mn)$ grid.
\end{proof}

\begin{figure}[htb]
  \centering
  \includegraphics[page=2]{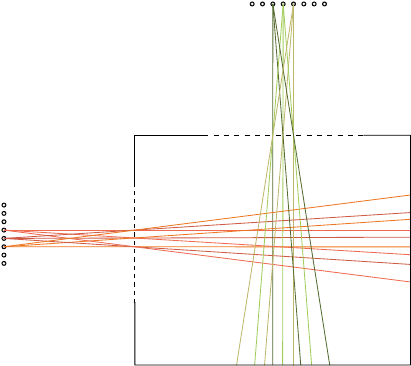}
  \caption{The visibility map of \T has complexity $\Omega(m^2n^2)$. Viewpoints are shown as white circles and rays indicate which part of
    the terrain is visible from the viewpoint.}

  \label{fig:2.5D-VisMap-quad}
\end{figure}


Clearly, this gives us also an $\Omega(m^2n^2)$ lower bound for the maximum complexity of the colored
visibility map $\colvis(\T,\G)$, and the Voronoi visibility map
$\vorvis(\T,\G)$.

Next, we show that the complexity of the colored visibility map
$\colvis(\T,\G)$ can be at most $O(m^2 n^2)$. To this end, we briefly summarize
a result on hidden surface removal~\cite{deBerg1996generalized}. Consider the
following problem: we are given a set $S$ of objects in $\R^3$, a viewpoint
$p$, and a light source $\ell$, and we wish to find the parts of $S$ that are
visible from $p$, partitioned into lit and unlit pieces. This information can
be captured in the \emph{generalized visibility map}: a subdivision of the
viewing plane (of $p$) into maximal regions such that each region is entirely
(i) invisible from $p$, (ii) visible from $p$ and lit by $\ell$, or (iii)
visible from $p$ but unlit by $\ell$. If $S$ is a 2.5D terrain with $n$
vertices, the generalized visibility map has complexity at most
$O(n^2)$~\cite{deBerg1996generalized}.

From this, we can easily derive that, for any pair of viewpoints $p$ and $q$, the colored
  visibility map $\colvis(\T,\{p,q\})$ has complexity $O(n^2)$: Let $p$ and $q$ be two viewpoints. We observe that the vertices of
  $\colvis(\T,\{p,q\})$ are either part of the generalized visibility map of \T
  with viewpoint $p$ and light source $q$, or they are invisible from $p$ and
  are part of the viewshed of $q$. Both the generalized visibility map, and
  viewshed $\V_q$ consist of at most $O(n^2)$ vertices.
	
Next, we introduce some terminology. Let $v$ be a vertex of \T, and let $p \in \G$ 
be a viewpoint.  We define the \emph {ray} of $p$ and $v$, denoted
$\uparrow_p^v$, to be the half line that starts at
$v$ and has vector $\overrightarrow{pv}$.  Similarly, let $p \in \G$ be a point and
$e=uv$ be an edge of $\T$.  The \emph {vase} of
$p$ and $e$, denoted \vase{p}{e}, is the region in $\R^3$ bounded by $e$,
\vase{p}{u}, and \vase{p}{v} (see Figure~\ref{fig:vases}). 

\begin{figure}[tb]
  \centering
  \subfigure[]{\includegraphics{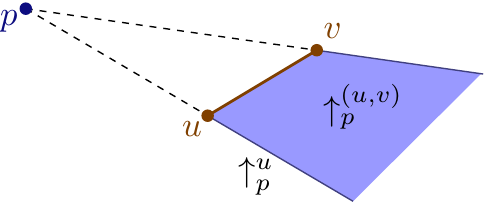}}
  \subfigure[]{\includegraphics{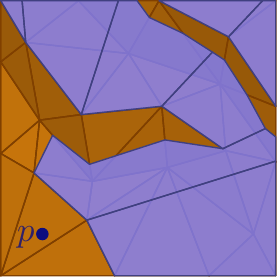}}
  \subfigure[]{\includegraphics{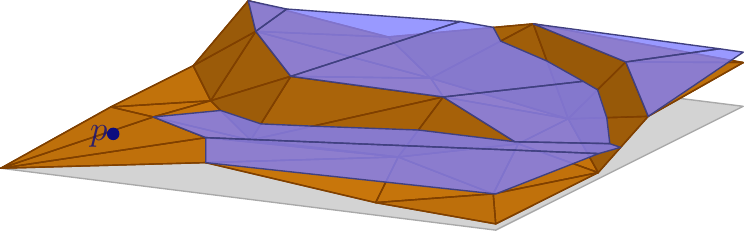}}
  \caption{(a) A ray and a vase. (b) The top-down view of a terrain \T with a
    single viewpoint $p$. The domain is decomposed in the viewshed $\vis(\T,p)$
    and a collection of vases. (c) a 3D view of \T and the vases of $p$. }
  \label {fig:vases}
\end{figure}

Clearly, the complexity of $\colvis(\T,\G)$ is bounded by its number of
vertices. We assume that the vertices are in general position, that is, there
are no four vertices co-planar. We now classify the vertices of $\colvis(\T,\G)$:

\begin {observation} \label {obs:vertex-types}
$\colvis(\T,\G)$ can have three types of vertices:
\begin {enumerate*}
 \item [(1)]  vertices of $\T$,
 \item [(2)]  intersections between an edge of $\T$ and a vase, and
 \item [(3)]  intersections between a triangle of $\T$ and two vases.
 \end {enumerate*}
\end {observation}

\begin{theorem}
  \label {thm:colviscomplexity2.5D}
  The colored visibility map $\colvis (\T,\G)$, for a 2.5D terrain \T, has
  complexity $O(m^2n^2)$.
\end{theorem}

\begin {proof}
  Each vase comes from a viewpoint in \G and an edge in $E(\T)$.  Clearly,
  $|V(\T)|$, $|E(\T)|$, $|F(\T)| \in O(n)$.  So, the number of vertices of type (1) is
  at most $O(n)$ and the number of vertices of type (2) is at most $O(mn^2)$. For each vertex $v$ of type (3), there exists a pair 
	of viewpoints $p$ and $q$ such that $v$ is also a vertex of $\colvis(\T,\{p,q\})$. Thus the number
of vertices of type (3) is at most $O(m^2n^2)$.
   We conclude that the
  complexity of the colored visibility map $\colvis(\T,\G)$ is at most
  $O(m^2n^2)$. 
\end{proof}

Since the vertices of the visibility map $\vis(\T,\G)$ are a subset of the
vertices of $\colvis(\T,\G)$, it then also follows that $\vis(\T,\G)$ has
complexity at most $O(m^2n^2)$.

\begin{theorem}
  \label{thm:viscomplexity2.5D}
  The visibility map $\vis (\T,\G)$, for a 2.5D terrain \T, has complexity
  $O(m^2n^2)$.
\end{theorem}

Finally, we are interested in the Voronoi visibility map.  $\vorvis(\T,\G)$ can
have vertices which are not in $\colvis(\T,\G)$. These vertices correspond to
intersections of Voronoi edges with terrain triangles. We use \emph{power
  diagrams} to show that the complexity can only be a factor $m$ higher than
that of $\colvis (\T,\G)$.

Let $\mathcal{C} = C_1,..,C_m$ be a set of $m$ circles in $\R^2$, and let $c_i$
and $r_i$ denote the center and radius of $C_i$, respectively. The (2D) power
diagram $\PD(\cal C)$ is the subdivision of $\R^2$ into $m$ regions, one for
each circle, such that $R_i = \{ x \in \R^2 \textrm{ s. t., for all } j \in
\{1,..,m\},\ \pow{C_i,x} \leq \pow{C_j,x} \}$, where $\pow{C_i,x} =
d_2(c_i,x)^2 - r_i^2$ (and $d_2(a,b)$ denotes the Euclidean distance between
$a$ and $b$ in $\R^2$). The (2D) power diagram of $m$ circles has complexity
$O(m)$ and can be computed in $O(m \log m)$ time
\cite{aurenhammer1987power,aurenhammer1988improved}.

Let $\VD(\G)$ denote the $3$-dimensional Voronoi diagram of $\G$. We observe
that the restriction of $\VD(\G)$ to any single plane $H$ in $\R^3$ corresponds
to a \emph{power diagram} $\PD(\mathcal{C}_\G)$ in $\R^2$: Assume without loss
of generality that $H$ is a horizontal plane at $z=0$, and let $\xi \geq
\max_{p \in \G} p_z^2$ be some large value. Any point $a \in H$ is closer to $p
\in \G$ than to $q \in \G$ if (and only if) $d(a,p) = d_3(a,p) \leq d_3(a,q)$,
and hence if $d_3(a,p)^2 \leq d_3(a,q)^2$. Using that $a_z = 0$ we can rewrite
this to $d_2(a,\underline{p})^2 - (\xi - p_z^2) \leq d_2(a,\underline{q}) -
(\xi - q_z^2)$. So, if we introduce a circle $C_p$ in $\mathcal{C}_\G$ for every
viewpoint $p$ with center $\underline{p}$ and radius $r_p$ such that $r_p^2 =
\xi - p_z^2$, then we get that $a$ is closer to $p$ than to $q$ if and only if
$\pow{C_p,a} \leq
\pow{C_q,a}$. 
Thus, we prove:

\begin {theorem}
  \label{thm:2.5d_complexity_vorvis}
  The Voronoi visibility map $\vorvis (\T,\G)$,  for a 2.5D terrain \T, has
  complexity $O(m^3n^2)$.
\end {theorem}

\begin{proof}
  The restriction of $\VD(\G)$ to any single plane in $\R^3$ corresponds to a
  power diagram in $\R^2$. A power diagram in $\R^2$ has complexity $O(m)$
  \cite{aurenhammer1987power}. The colored visibility map $\colvis(\T,\G)$ has
  $O(m^2n^2)$ vertices, and hence $O(m^2n^2)$ faces. Thus, each face of
  $\colvis(\T,\G)$ can contain at most $O(m)$ vertices of $\vorvis(\T,\G)$. 
\end{proof}


\subsection{Algorithms to Construct the Visibility Structures}

\subsubsection{Computing the (Colored) Visibility Map}

Katz et al.~\cite{kos-ehsrosus-92} developed an $O((n \alpha(n) + k) \log n)$ time
algorithm to compute the viewshed of a single viewpoint, where $k$ is the
output complexity and $\alpha(n)$ is the extremely slowly growing inverse of
the Ackermann function.  Coll et al.~\cite{cms-gvmv-10} use this algorithm to
compute the visibility map of a 2.5D terrain in $O(m^2n^4)$ time and
space.

\begin{figure}[tb]  
\centering
  \includegraphics{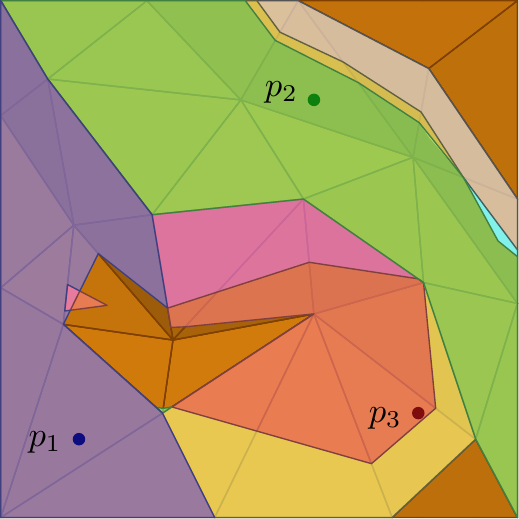}
  \caption{Overlay \overviewshed.}
  \label{fig:overlay_viewsheds}
\end{figure}

Essentially they project the individual viewsheds onto $\R^2$, and
construct the overlay $\overviewshed = \bigoplus_{p \in \G}
\underline{\viewshed{p}}$ (see Figure~\ref{fig:overlay_viewsheds}). It is then
easy to construct the (colored) visibility map from \overviewshed. We use the
same approach. However, using our observations from the previous section, we can
show that even if the viewsheds have complexity $\Theta(n^2)$, we can compute the
(colored) visibility map in $O(m^3n^2)$ time. More specifically, we obtain:

 \begin{lemma}
  \label{lem:construction_overlay}
  Given a 2.5D terrain \T,
  the planar subdivision \overviewshed can be constructed in $O(m(n\alpha(n) +
  \min(k_c,n^2))\log n + mk_c)$ time.
\end{lemma}

\begin{proof}
  Let $\G = p_1,..,p_m$ be the set of viewpoints. To compute the individual
  viewsheds we use the algorithm of Katz et al. \cite{kos-ehsrosus-92}. For a
  single viewpoint $p_i$ this takes $O((n\alpha(n) + \kappa_i)\log n)$ time,
  where $\kappa_i$ is the complexity of the viewshed \viewshed{p_i}. Hence, we
  need $O(\sum_{i=1}^m (n\alpha(n) + \kappa_i)\log n) = O(mn\alpha(n)\log n +
  \log n \sum_{i=1}^m \kappa_i)$ time to compute all $m$ viewsheds. Each
  individual viewshed has maximum complexity $O(n^2)$. We also have that for
  all $i$, $\kappa_i \leq k_c$. Hence, $\sum_{i=1}^m \kappa_i =
  O(m\min(k_c,n^2))$.

  We then iteratively build \overviewshed, starting with the empty
  subdivision. In each step $i$, with $1 \leq i \leq m$, we compute the overlay
  of \overviewshed with $\underline{\viewshed{p_i}}$ using the algorithm by
  Finke and Hinrichs \cite{finke1995overlaying}. Given two simply connected
  planar subdivisions $S_1$ and $S_2$ with $n_1$ and $n_2$ vertices
  respectively, this algorithm computes the overlay $S_1 \oplus S_2$ in
  $O(n_1+n_2 + k)$ time and space, where $k$ denotes the output
  complexity. 
  Since at every step the complexity of \overviewshed is at most $k_c$, it
  follows that we can compute \overviewshed in $O(mk_c)$ time. The lemma
  follows.
\end{proof}

\begin{theorem}
  \label{thm:2.5d_compute_vis_and_colvis}
Given a 2.5D terrain \T, both the visibility map $\vis(\T,\G)$ and the colored visibility map
  $\colvis(\T,\G)$  can be constructed
  in $O(m(n\alpha(n) +  \min(k_c,n^2))\log n + mk_c)$ time.
\end{theorem}

\subsubsection{Computing the Voronoi Visibility Map.}

Let $F$ be a face of the colored visibility map $\colvis(\T,\G)$, and let
$\G_F$ denote the set of viewpoints that can see $F$. For each such face $F$ we
compute the intersection of $F$ with the $\VD(\G_F)$. We do this via the power
diagram: i.e. consider the plane $H$ containing $F$, and compute the power
diagram on $H$ with respect to the viewpoints in $\G_F$. This takes
$O(k_cm\log m)$ time in total, since $\colvis(\T,\G)$ has $O(k_c)$ faces, and
each power diagram can be computed in $O(m \log m)$ time. Each power diagram is
constrained to a single face, so we glue all of them together and project the
result onto $\R^2$. This yields a subdivision
$\mathbb{W}$ of size $O(k_cm)$. We now compute \overviewshed in $O(m(n\alpha(n) +
  \min(k_c,n^2))\log n + mk_c)$ time (as described above), and overlay it with
$\mathbb{W}$ in $O(k_cm + k_c + k_v) = O(k_cm)$ time. Hence:

\begin{theorem}
  \label{thm:2.5d_compute_vorvis}
Given a 2.5D terrain \T, the Voronoi visibility map $\vorvis(\T,\G)$ can be
  constructed in $O(m(n\alpha(n) + \min(k_c,n^2))\log n + mk_c\log m)$ time.
\end{theorem}


\section{Viewpoints with limited sight}
\label{sec:limited-sight}

Throughout the paper we have assumed viewpoints have unlimited sight, as is common in computational geometry.
However, it is interesting to study what happens if each viewpoint $p_i$
can only see for a limited distance $r_i$, since this is reasonable for many applications.
Below we explain how our results can be adapted to this realistic setting.

\subsection {1.5D terrains}

We study the case $r_1=\cdots=r_m$.
The maximum complexity of the colored and Voronoi visibility maps
remains $\Theta(m n)$. In contrast, reducing the range of the
viewpoints can worsen the maximum complexity of the visibility
map, and in particular it is possible to achieve complexity
$\Theta(m n)$.

Computing the visibility map seems more expensive in this case. In particular, dividing
the map into left- and right- visibility might not result in an
output sensitive algorithm, since the left-visibility map can have much higher
complexity than that of the final map. 
Unfortunately, it is unclear how to produce the visibility map in an output-sensitive fashion.
Therefore we resort to extracting the visibility map from the colored visibility map.

The colored visibility map can be constructed in almost the same running time as when the range of the viewpoints is infinite, but the algorithm in Section~\ref{subsec:1.5-alg-col} needs several modifications, which are explained below.
Finally, the algorithm for the Voronoi visibility map presented in Section~\ref{sec:alg1.5Vor} works with no modification if viewpoints see only up to a limited distance, since the colored visibility map encapsulates all events related to the appearance or disappearance of viewpoints, regardless of whether sight is unlimited or not.

\subsubsection{Construction of  $\colvis(\T,\G)$} \label{subsec:colvis-limitedsight}

We next show that, if the viewpoints have the same limited sight, the colored visibility map $\colvis(\T,\G)$, for a 1.5D terrain  \T,  can be constructed in $O(n \log n + (k_c+m^2) \log n)$ time.

The following algorithm first constructs a map for left-visibility while sweeping the terrain from left to right, then a map for right-visibility in a symmetric way, and finally merges both to obtain $\colvis(\T,\G)$. 

The algorithm maintains a collection of \emph{active} viewpoints \A, defined as viewpoints that are both visible and within sight at that point of the sweep.
\A is represented with a balanced binary tree where the viewpoints are sorted by their $x$-coordinate. 
In addition, instead of using one global event queue, we will maintain one event queue for each edge of the terrain, where all events that take place on points of that edge will be stored.

In the following, assume that each viewpoint is the center of a disk of radius $r$.

The following observation is somewhat analogous to the order claim for unlimited sight.

\begin{observation}
\label{obs:left-visLimSight}
Let $p_i$ and $p_j$ be two viewpoints that can see within distance $r$, and suppose $p_i$ is to the left of
$p_j$. Let $q \in \T$ be a point on the terrain that has an unobstructed line of sight to both $p_i$ and $p_j$, and such that $q$ is at distance $r$ from $p_j$. 
Then if $q$ is on the bottom right quadrant of $p_j$,  $q$ is not within sight of $p_i$.
\end{observation}


The sweep starts at the leftmost edge of the terrain.
We distinguish the following four types of events:
(i) a new edge begins;
(ii) an active viewpoint stops seeing the current edge from its upper right quadrant, that is, the edge crosses the boundary of the upper right quadrant of the disk of radius $r$ centered at the viewpoint (see Figure~\ref{fig:lim-range-colvis}, left);
(iii) an active viewpoint stops seeing the current edge from its bottom right quadrant, that is, the edge crosses the boundary of the bottom right quadrant of the disk of radius $r$ centered at the viewpoint; and 
(iv) a viewpoint becomes active.

\begin{figure}[tb]
\centering
\includegraphics[scale=1.2]{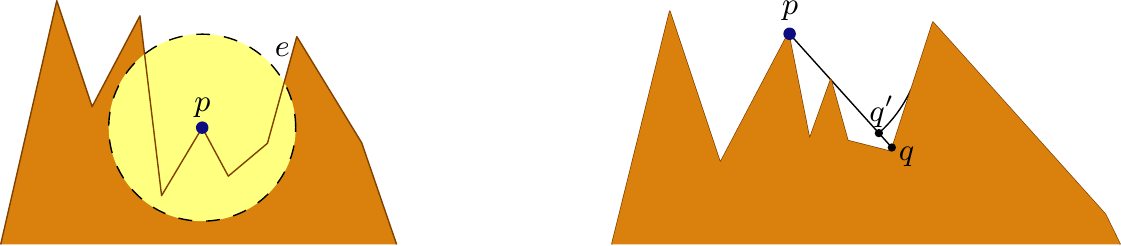}
\caption{Left: The active viewpoint $p$ stops seeing edge $e$ from its upper right quadrant. Right: Computation of the first point of the terrain where viewpoint $p$ reappears.}
\label{fig:lim-range-colvis}
\end{figure}

Next we explain how the algorithm handles each type of event.

(i) A new edge $v_iv_j$ begins.

Several viewpoints may stop being visible after $v_i$ and therefore need to be identified and removed from \A. However, before removing the viewpoints we need to check where they reappear and add at those points events of type (iv) in the event queues of the corresponding edges. The computation of the points of the terrain (if any) where they reappear is done as follows: First, for every viewpoint $p$, we perform a ray-shooting query with the ray $\overrightarrow{pv_i}$. Let $q$ be the point obtained (if any). If $q$ is at distance at most $r$ from $p$, we return $q$. Otherwise, let $q'$ be the point on $\overrightarrow{pv_i}$ at distance $r$ from $p$. We perform a ray-shooting query with a circular arc of radius $r$ going counterclockwise from $q'$ (see Figure~\ref{fig:lim-range-colvis}, right). We return the point that we obtain (if any).

Note that the viewpoints that cannot see $v_iv_j$ have consecutive $x$-coordinates among those in \A, must include the rightmost viewpoint, and thus appear on the ``right side'' of \A.  
Therefore, it is possible to process all viewpoints from right to left until the first viewpoint that sees $v_iv_j$ is found. 

In case there is a viewpoint on $v_i$, it should be inserted in \A as the rightmost element. Furthermore, we also need to find the point (if any) where this viewpoint stops seeing the terrain from its upper right quadrant. This can be done using a circular ray-shooting query with a circular arc of radius $r$ going clockwise from the point that has the same $x$-coordinate as $v_i$, lies above $v_i$, and is at distance $r$ from $v_i$. If the ray-shooting query returns a point in the upper right quadrant, then this point becomes an event of type (ii) on the appropriate edge.
 
Furthermore, the viewpoints that see $v_i$ may not see the whole edge and also need to be identified. 
At this stage we only need to identify the viewpoints that stop seeing $v_iv_j$ from their bottom right quadrant. 
According to Observation~\ref{obs:left-visLimSight}, \A can be searched from left to right until the first viewpoint that fully sees $v_iv_j$ is found. For all the viewpoints that stop seeing $v_iv_j$ from their bottom right quadrant, we add events of type (iii) in the event queue of edge $v_iv_j$. 


(ii) An active viewpoint stops seeing the current edge from its upper right quadrant.

Such a viewpoint will never become active again during the sweep, so it is permanently removed from \A.

(iii) An active viewpoint stops seeing the current edge from its bottom right quadrant.

Even though the viewpoint no longer sees the current edge, it may become active again during the sweep. Therefore, we find the next point where the viewpoint sees the terrain again using a circular ray-shooting query with a circular arc of radius $r$ going counterclockwise from the point where the current edge crosses the boundary of
the bottom right quadrant of the disk of radius $r$ centered at the viewpoint. If such a point exists, an event of type (iv) is added on the corresponding edge. Note that if the ray-shooting hits a point of the terrain on the right of the current point, then such point is visible, since the sector of the disk between the current point and the re-entry point is empty. Finally, we remove the viewpoint from \A.

(iv) A viewpoint becomes active.

We insert the viewpoint in \A and check whether the viewpoint stops seeing the current edge from its bottom right quadrant. In case it does, an event of type (iii) is added at the point where the viewpoint becomes invisible again.

Once the left-visibility map is constructed, the right-visibility map can be built by mirroring the procedure described above. In the end, both maps are merged and the resulting map is $\colvis(\T,\G)$.

\paragraph{Running time.} The events of type (i) corresponds to a terrain vertex, thus there are $O(n)$ of them, while each event of type (ii), (iii) or (iv) can be associated with a change in the final colored visibility map.
Furthermore, the algorithm only processes viewpoints whose active status changes.
 
Processing an event involves a constant number of rectilinear or circular ray-shootings queries, taking $O(\log n)$ time each~\cite{cegghss-rspugt-94,cceo-hdcrs-04}, a constant number of operations in \A and in some of the event queues of the edge, taking $O(\log m)$ time each. 

Events of type (ii), (iii) and (iv) only involve a single viewpoint and always produce a change in $\colvis(\T,\G)$. 
Events of type (i) may involve many viewpoints and produce only one change in $\colvis(\T,\G)$. 
In the same way as for unlimited sight, we bound the time spent on these events by charging the time to the point at which the viewpoint becomes active again---events of type (iv). 
As before, it can happen that several viewpoints become active exactly at the same point, but for a given pair of viewpoints, this happens at most once. 
Therefore, the total time spent handling all events is $O((k_c+m^2) \log n)$. 

Finally, the terrain has to be preprocessed to allow fast ray-shooting queries with rectilinear or circular arcs, a procedure requiring $O(n \log n)$ time and $O(n)$ space~\cite{cegghss-rspugt-94,cceo-hdcrs-04}. 

It follows that the total running time of the algorithm is $O(n \log n + (k_c+m^2) \log n)$.

\paragraph{Correctness.} The algorithm explained above is correct once it is proven that it cannot miss any changes in the final map. The sweep should not miss when a viewpoint becomes within sight or out of reach, nor when it appears or disappears on the terrain. At every vertex of the terrain, the algorithm verifies which viewpoints remain visible and within reach, as well as which are hidden, so it is impossible to miss that a viewpoint is no longer active. 
Moreover, for each viewpoint that no longer is active the algorithm computes the next point where it becomes active again (if such point exists). 
Regarding range, a viewpoint can stop seeing the terrain from its bottom left/right quadrant or from its upper left/right quadrant, but each viewpoint only stops seeing the terrain from its upper left/right quadrant once. This exit point is found when the sweep stops at the viewpoint for the first time, thus this type of event cannot be missed by the algorithm. The exit points from the bottom left/right quadrants are found when the sweep reaches the edges where these events take place. 
According to Observation~\ref{obs:left-visLimSight}, it suffices to check the left end of \A, since the viewpoints stop seeing the terrain in order, as long as the exit is made through their bottom left/right quadrant. In conclusion, these exit points are not missed either.

We summarize the results in this section with the following theorem.

\begin{theorem} \label{thm:1.5-map-alg-lim}
Given a 1.5D terrain \T and a set $\G$ of viewpoints with the same limited sight,
$\vis ({\cal T},\G)$ and $\colvis ({\cal T},\G)$ can be constructed in $O(n \log n + (m^2 + k_c )\log n)$ time.
$\vorvis({\T},\S)$ can be
  constructed in $O(n \log n + (m^2 + k_c) \log n + k_v (m+\log n \log m))$ time.
\end{theorem}

\subsection{2.5D terrains}

Let $B_i$ denote the ball of radius $r_i$ centered at $p_i$, and let $\B(\G)
= \{B_i \mid p_i \in \G\}$. We are now interested in the complexity of
$\vis(\T,\G) \cap \bigcup \B(\G)$. 

In this case, limited sight brings additional problems as it creates several new types of vertices, which are not contemplated in Obs.~\ref {obs:vertex-types}:
\begin {enumerate*}
 \item [(4)] intersections between an edge of $\T$ and
a single sphere,
\item[(5)] intersections between a triangle of $\T$ and two spheres, and
\item[(6)] intersections between a triangle of $\T$, a sphere and a vase.\footnote{There are more types of ``vertices'' that do not lie on \T, for example,
  intersections of three spheres, intersections of two spheres and a vase,
  etc.}
 \end {enumerate*}
  The number of vertices of types (4), (5), and (6) is $O(mn)$,  $O(m^2n)$, and
$O(m^2n^2)$ respectively. Since all these types of vertices are dominated by the number of vertices of type
(3), the complexity bounds of our visibility structures do not change.


\section{Final Remarks} \label{sec:final-remarks}
In this paper we studied visibility with multiple viewpoints on polyhedral terrains for the first time,
 and presented a thorough study on the complexities and algorithms for three fundamental visibility structures.
Our results show that considering multiple viewpoints converts classical visibility problems into much more challenging ones, even for 1.5D terrains.

Moreover, our results lead to many intriguing questions.  For 1.5D terrains, is
there an efficient algorithm to construct the Voronoi visibility map whose
running time does not depend on $k_c$?  In 2.5D, the worst-case complexity for
the Voronoi visibility map is not tight; it would be 
interesting to close this gap. Algorithmically, in 2.5D the main challenge is
to find an algorithm to construct the structures directly, avoiding the
computation of the individual viewsheds.

\paragraph{Acknowledgments.}
The authors would like to thank several anonymous reviewers for their helpful and insightful comments, and Mark de Berg for pointing us to~\cite {deBerg1996generalized}.
 F.~H., V.~S., and R.~S. were partially supported by
projects MINECO MTM2012-30951/FEDER, Gen. Cat. DGR2009SGR1040, and ESF EUROCORES programme EuroGIGA
- ComPoSe IP04 - MICINN Project EUI-EURC-2011-4306. 
M.~L. and
F.~S. were supported by
Netherlands Organisation for Scientific Research (NWO) under project no 639.021.123 and 612.001.022. I.~M. was supported by FCT project PEst-C/MAT/UI4106/2011 with COMPETE number FCOMP-01-0124-FEDER-022690. M.~S.\ was supported by the
project NEXLIZ - CZ.1.07/2.3.00/30.0038, which is co-financed by
the European Social Fund and the state budget of the Czech
Republic, by ESF EuroGIGA project ComPoSe as
F.R.S.-FNRS - EUROGIGA NR 13604, and by ESF EuroGIGA project GraDR as
GA\v{C}R GIG/11/E023. 
R.~S. was funded by Portuguese funds through CIDMA (Center for Research and Development in Mathematics and Applications) and FCT (Funda\c{c}{\~a}o para a Ci{\^e}ncia e a Tecnologia), within project PEst-OE/MAT/UI4106/2014, and by FCT grant SFRH/BPD/88455/2012. 
\bibliographystyle{abbrv}
\bibliography{refs}

\end{document}